\documentclass[aps,prb,twocolumn,superscriptaddress,floatfix,nofootinbib]{revtex4-2}
\usepackage{amssymb}
\usepackage{physics}
\usepackage{graphicx}
\usepackage{times}
\usepackage{amsmath}
\usepackage{amsthm}
\usepackage{amsfonts}
\usepackage{comment}
\usepackage{mathtools}
\usepackage[T1]{fontenc}
\usepackage[latin9]{inputenc}
\usepackage{array}
\usepackage{multirow}
\usepackage{color}
\usepackage{esint}
\usepackage{bm}
\usepackage{color}
\usepackage{xcolor}
\usepackage{bbm}
\usepackage{hyperref}
\usepackage{babel}
\usepackage{soul}
\usepackage[normalem]{ulem}
\usepackage{booktabs}
\usepackage{comment}

\newtheorem{thm}{Theorem}[section]
\newtheorem{lemma}[thm]{Lemma}

\theoremstyle{definition}
\newtheorem{definition}[thm]{Definition}

\newcommand{\eqnref}[1]{Eq.~\eqref{#1}}

\usepackage{float}
\hypersetup
{	colorlinks,%
	citecolor=green,%
	linkcolor=red,%
	urlcolor=blue%
}

\newcommand{\unit}{\mathbbm{1}}
\newcommand{\Htot}{\mathcal{H}}
\newcommand{\Hclock}{\mathcal{H}_{\text{clock}}}
\newcommand{\Hspin}{\mathcal{H}_{\text{spin}}}

\newcommand{\propagator}[2]{U_{#1:#2}}
\newcommand{\SecRenyi}[2]{S^{(2)}_{#1}(#2)}

\begin{document}
\title{Infinite temperature at zero energy}
\author{Matteo Ippoliti}
\email{ippoliti@utexas.edu}
\affiliation{Department of Physics, The University of Texas at Austin, Austin, Texas 78712, USA}
\author{David M. Long}
\email{dmlong@stanford.edu}
\affiliation{Department of Physics, Stanford University, Stanford, California 94305, USA}

\begin{abstract}
    We construct a family of static, geometrically local Hamiltonians that inherit eigenstate properties of periodically-driven (Floquet) systems. Our construction is a variation of the Feynman-Kitaev clock---a well-known mapping between quantum circuits and local Hamiltonians---where the clock register is given periodic boundary conditions.
    Assuming the eigenstate thermalization hypothesis (ETH) holds for the input circuit, our construction yields Hamiltonians whose eigenstates have properties characteristic of infinite temperature, like volume-law entanglement entropy, across the whole spectrum---including the ground state.
    We then construct a family of exactly solvable Floquet quantum circuits whose eigenstates are shown to obey the ETH at infinite temperature. Combining the two constructions yields a new family of local Hamiltonians with provably volume-law-entangled ground states, and the first such construction where the volume law holds for all contiguous subsystems.
\end{abstract}

\maketitle

\section{Introduction}
    \label{sec:Intro}

    Tools from quantum information have found increasing application in many-body physics, both in and out of equilibrium. 
    Many processes are now understood in terms of \emph{entanglement entropy}---an observable-agnostic quantification of correlations.
    The spatial scaling of entanglement entropy sharply distinguishes different types of states. 
    At high energy density, the eigenstate thermalization hypothesis (ETH) characterizes eigenstates as effectively random, and relates their entanglement entropy to the thermodynamic entropy. This gives rise to extensive (``volume law'') scaling of entanglement entropy~\cite{Deutsch1991statmechclosed,Srednicki1994chaosthermalization,Srednicki1999approach,Rigol2008thermalization,D'Alessio2016ETHreview,Deutsch2018ETHreview}.
    At low energies, different structures of entanglement are possible. Gapped ground states obey the ``area law'' (proven in one dimension~\cite{Hastings2007arealaw} and for some models in higher dimension~\cite{Anshu2022arealaw}): the entanglement entropy of a subsystem scales in proportion to its boundary, representing a limited amount of correlations.
    Gapless ground states, on the contrary, can violate the area law---for example, one-dimensional quantum critical states exhibit logarithmic scaling of entanglement~\cite{Calabrese2004qft,Calabrese2009cft}.

    Surprisingly, violations of the area law in gapless ground states can be maximal: families of local, one-dimensional spin Hamiltonians were constructed which have volume-law entangled ground states~\cite{Gottesman2010entanglementvsgap,Irani2010groundstateentanglement,Vitagliano2010volumelaw,Ramirez2014volumelaw}, notably including translationally invariant models known as (colored, area-deformed) Motzkin and Fredkin chains~\cite{Bravyi2012criticality,Movassagh2016supercritical,Zhang2017extensive,Salberger2018fredkin}.
    This shows that ground states can have an amount of entanglement comparable to that of highly excited states, challenging prior beliefs about the nature of quantum correlations at low energy.
    At the same time, despite their high entanglement, these ground states are very different from genuine ETH-obeying, highly excited eigenstates. The structure of their parent Hamiltonians, which enables their exact solution, also imparts a ``rainbow'' structure to the state~\cite{Alexander2019rainbowtn,Alexander2021holographictn}: the entanglement is dominantly shared between pairs of qubits related by a reflection about the mid-point of the chain. While the half-chain entropy is high, local reflection-symmetric subsystems can have much lower entanglement. 
    Recent work~\cite{Balasubramanian20232dentanglement} extended the Motzkin construction to higher dimensions. While at present there is no ``rainbow''-like picture for the resulting states, there is also no proof of volume-law entanglement for arbitrary local subsystems.
    It thus remains unclear how closely the entanglement structure of a (gapless) ground state can resemble that of a highly excited state, with a thermal, ``random-looking'' character and in particular volume-law scaling for all local subsystems.

    In this work,
    we construct a large family of geometrically local static Hamiltonians which display properties associated to infinite temperature across almost all eigenstates. In particular, the eigenstates---including the ground state---are volume-law entangled for all contiguous subsystems.
    Our construction is based on the Feynman-Kitaev clock (FK clock) model~\cite{Kitaev2002book}, a well-known model which embeds the output of a quantum circuit into the ground state of a static Hamiltonian.
    By modifying the FK clock model to include periodic boundary conditions---so that the final ``time'' on the clock is the same as the initial time---we embed \emph{Floquet states}~\cite{Bukov2015highfreq,Rudner2020floquetengineershandbook} of an input staircase circuit in eigenstates of the Hamiltonian.
    The model is illustrated qualitatively in Fig.~\ref{fig:idea}. It is a two-leg ladder with geometrically local interactions; its ground state is built from a Floquet state on one leg and a conventional ground state on the other, coupled together by a fixed entangler circuit.
    As almost all Floquet states of generic quantum circuits are expected to obey the ETH at infinite temperature, 
    this structure brings infinite-temperature features to a local Hamiltonian ground state.
    We prove in particular that it implies volume-law entanglement of the clock Hamiltonian eigenstates, including the ground state.

    While eigenstates of generic circuits are widely believed to obey the ETH, in order to remove this physical assumption, we construct a \emph{specific} family of Floquet circuits whose eigenstates provably obey ETH at infinite temperature.
    Used as an input for the FK clock Hamiltonian, this produces ground states which are provably volume-law entangled.
    These circuits are a quantum implementation of linear feedback shift registers (LFSRs)~\cite{Tausworthe1965random,Massey_1969_shiftregister,Klein2013lfsr}---classical circuits that play a role in pseudorandom number generation.
    Independently of their application to our construction, we expect these Floquet models to be very useful in the study of ergodicity, thermalization, and chaos in quantum systems~\cite{D'Alessio2016ETHreview,Deutsch2018ETHreview}.
    While there are several solvable models of thermalization based on random~\cite{Chan2018randomfloquet,Fisher2023randomcircuitreview,Liao2022fieldtheoryapproacheigenstate} or dual unitary~\cite{Bertini2025exactlysolvablemanybodydynamics} circuits, and proofs of thermalization in a large class of translationally-invariant Hamiltonians~\cite{Huang2020instabilitylocalizationtranslationinvariantsystems,Pilatowskycameo2025thermalization}, this solvability does not extend to a closed-form expression for the eigenstates, and only a proof of a weak form of ETH~\cite{Biroli2010rarefluctation,Mori2016weakeigenstatethermalizationlarge,Lydzba2024normalweaketh,Vikram2025bypassingeigenstatethermalizationexperimentally}.
    The only class of prior models known to us in which ETH is provably obeyed are random matrices~\cite{Potters2020rmtbook,Khaymovich2020lnrp,Sugimoto2023ethmeanfield}, which have no locality structure.
    Further, the LFSR dynamics conform with many statements of quantum chaos, but fail many others.
    Indeed, while we show that their eigenstates obey ETH, in the sense that off-diagonal matrix elements of local operators are small, their spectra are nothing like those of a Haar random unitary.
    As an exactly solvable edge case, they clarify the implications of different criteria of chaos and thermalization, and should be broadly useful in sharpening our understanding of many-body dynamics.
    We note that LFSR quantum circuits also recently appeared in Ref.~\cite{Kim2025catalyticzrotationsconstanttdepth} in an unrelated context.

    \begin{figure}
        \centering
        \includegraphics[width= 0.99\columnwidth]{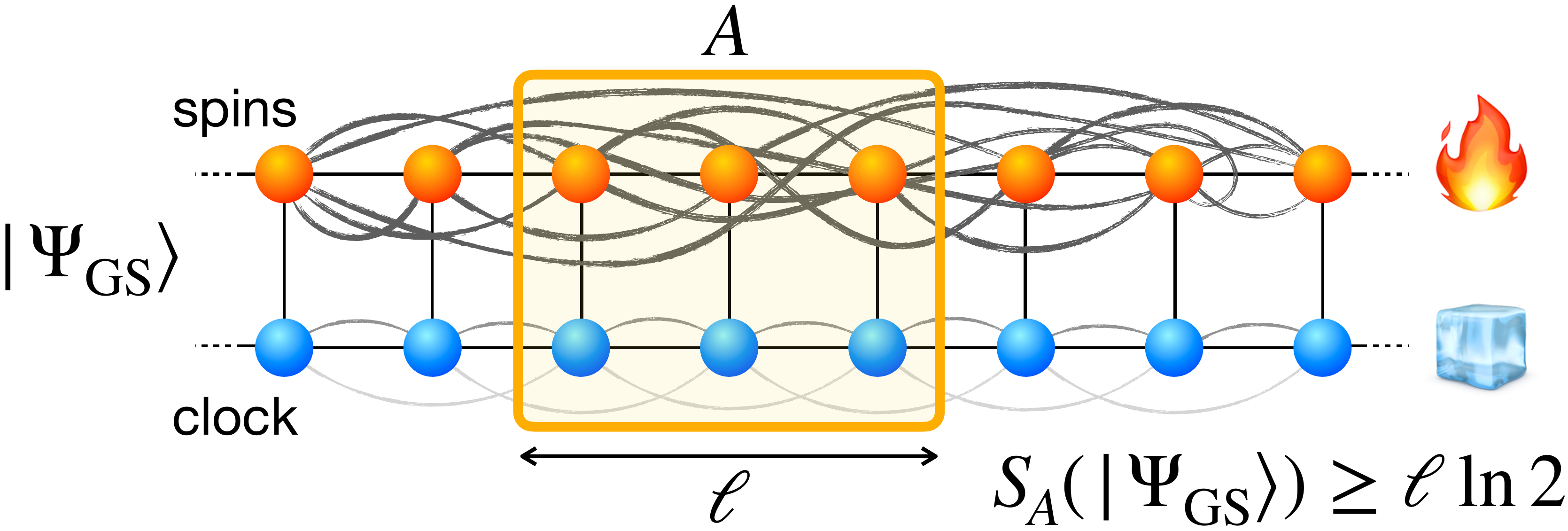}
        \caption{Schematic of the ground state $\ket{\Psi_\mathrm{GS}}$ of the periodic FK clock Hamiltonian. The system is a two-leg ladder comprising a qubit chain (``spins'', top) and a fermionic chain (``clock'', bottom), with geometrically local interactions (4-body couplings around plaquettes). The ground state is built by entangling in a specific way an infinite-temperature Floquet state on the spins and a zero-temperature state on the clock.
        Wavy lines qualitatively represent entanglement. Any subinterval $A$ of the spin-clock ladder (yellow box) is volume-law entangled, with $S_A(\ket{\Psi_\mathrm{GS}}) \geq \ell \ln 2$, $\ell$ being the length of the interval.
        \label{fig:idea}}
    \end{figure}

    The rest of the paper is organized as follows.
    In \autoref{sec:FKClockModel} we review the standard Feynman-Kitaev clock construction, introduce its modifications with periodic boundary conditions on the clock, and prove that its eigenstates exhibit volume law entanglement whenever a corresponding eigenstate of the input circuit obeys the ETH (Theorem~\ref{thm:volumelaw_1hand}).
    Then, in \autoref{sec:LFSRModel}, we introduce Floquet quantum circuits based on classical LFSRs and prove that almost all eigenstates thereof obey the ETH at infinite temperature (Theorem~\ref{thm:eth}), which in combination with the previous result yields a rigorous construction of local Hamiltonians with volume-law entangled ground states (Theorem~\ref{thm:volumelaw_gs}).
    We discuss the implications and possible extensions of our work in \autoref{sec:Discussion}.

\section{Feynman-Kitaev Clock Model}
    \label{sec:FKClockModel}

    \begin{figure*}
        \centering
        \includegraphics[width=\textwidth]{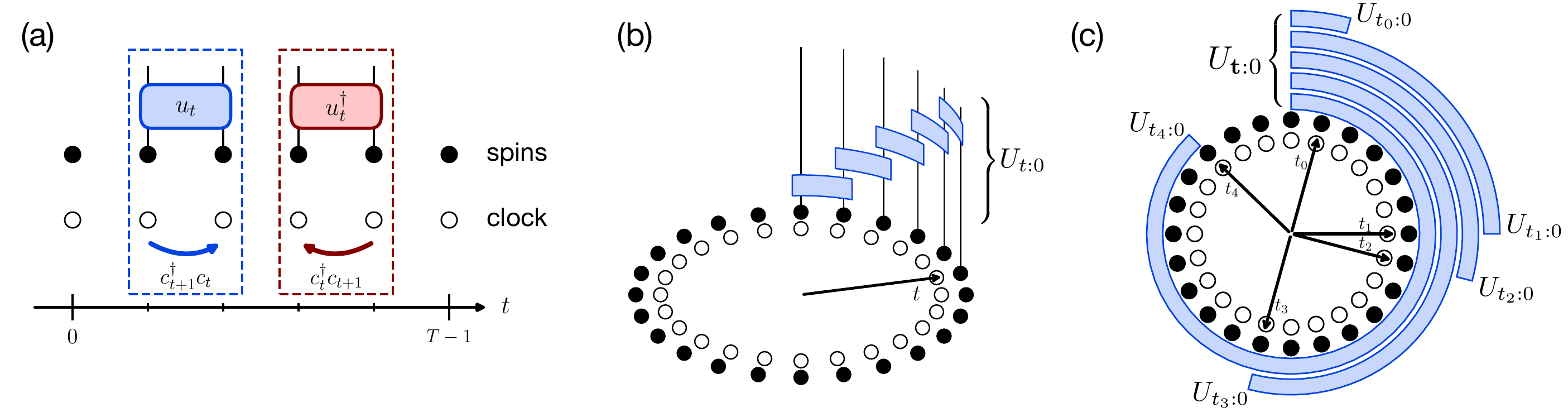}
        \caption{Schematic of the periodic clock Hamiltonian construction. 
        (a) The Hamiltonian acts on a chain of $n$ qubits (\emph{spins}, solid dots) coupled to a fermionic tight-binding chain (\emph{clock}, open dots) with sites enumerated by $t \in [0,T-1]$. In the figure we take $T = n$. Dashed boxes illustrate the local Hamiltonian terms: the clock particle, or \emph{hand}, hops forward (resp. backward) and a local gate $u_t$ (resp. $u_t^\dagger$) is applied to the spin system. 
        (b) Periodic clock: the hand can hop between $T-1$ and 0. If the spin system is in an eigenstate of $U_{T:0}$, then the hand is always trailed by a string operator $\propagator{t}{0} = u_{t-1}u_{t-2}\cdots u_0$ acting on the spin system, even when it winds around the circumference nontrivially. 
        (c) Many-handed periodic clock. Each of the $M$ hands (here $M=5$) is trailed by a string operator $U_{t_i:0}$, with the longest string acting first.
        }
        \label{fig:clocks}
    \end{figure*}
    
    Our construction of volume-law-entangled ground states is based on the Feynman-Kitaev clock (FK clock) model~\cite[Chapter~14]{Kitaev2002book}.
    This model is designed to embed the result of a quantum computation in a Hamiltonian ground state. In its original context, it was intended to show that finding the ground state of a \(k\)-local Hamiltonian is, in general, hard. (It is QMA complete~\cite{Kitaev2002book}.)

    By introducing periodic boundary conditions to the FK clock model, we demonstrate that it can be used to embed any Floquet eigenstate of a quantum circuit into the ground state of a local Hamiltonian.

    In \autoref{subsec:OpenClock} we review the FK clock construction, which takes as input a quantum circuit, and produces a Hamiltonian whose ground state contains the output of the quantum circuit. 
    Then we introduce periodic boundary conditions in \autoref{subsec:PeriodicClock}, and demonstrate that the new Hamiltonian eigenstates are related to Floquet states of the input circuit.
    In \autoref{subsec:MBClock}, we extend the clock model to a geometrically local many-body Hamiltonian.
    In \autoref{subsec:VolumeLawEEE} we prove the main result of the section: that eigenstates of the clock model have volume-law entanglement provided the Floquet states of the input circuit obey (diagonal) ETH. 
    Finally, in \autoref{subsec:SpinExpectations}, we discuss the specific structure of local expectation values in the FK clock ground state, justifying the cartoon of Fig.~\ref{fig:idea}.

    \subsection{Open clocks}
        \label{subsec:OpenClock}
        
        The FK clock model is defined on a tensor product Hilbert space
        \begin{equation}
            \Htot = \Hspin \otimes \Hclock
        \end{equation}
        with factors $\Hspin$ and $\Hclock$ of dimension $2^n$ and $T$ respectively. $\Hspin$ represents a system of $n$ qubits (referred to as spins) while $\Hclock$ contains states of an auxiliary \emph{clock} degree of freedom.
        Denote an orthonormal basis for \(\Hclock\) by \(\ket{t}\) with \(t \in \{0, ..., T-1\}\). Given any quantum circuit consisting of \(T-1\) local gates on \(\Hspin\), \(\{u_t\}_{t=0}^{T-2}\), the FK clock Hamiltonian is
        \begin{equation}\label{eqn:OpenClock}
            H = \Pi^\text{(init)} + \sum_{t=0}^{T-2} \Pi_t^\text{(tick)},
        \end{equation}
        where
        \begin{subequations}
        \begin{align}
            \unit - \Pi^\text{(init)} =& \ketbra{\psi_0}{\psi_0} \otimes \ketbra{0}{0},\\
            \unit - \Pi_t^\text{(tick)} =&  \frac{1}{2} \big[ u_t \otimes \ket{t+1} - \unit \otimes \ket{t} \big] \nonumber \\
            &\quad \times \big[ u^\dagger_t \otimes \bra{t+1} - \unit \otimes \bra{t} \big],
        \end{align}
        \end{subequations}
        and \(\ket{\psi_0}\) is a state in \(\Hspin\).
        Note that the terms in the Hamiltonian act locally on the spin system (the unitaries $u_t$ appearing in $H$ are {\it geometrically local} gates).
        We refer to this model as the open FK clock model, to indicate that it has open boundary conditions---there are no terms in the Hamiltonian which couple \(\ket{T-1}\) to \(\ket{0}\).

        The interpretation of the Hamiltonian~\eqref{eqn:OpenClock} is as follows.
        The projector \(\Pi^\text{(init)}\) picks out an initial state \(\ket{\psi_0}\) on which some quantum computation \(\prod_{t=0}^{T-2} u_t\) is to be performed.
        Then each of the projectors \(\Pi_t^\text{(tick)}\)tries to connect each ``tick'' of the clock with the application of a gate $u_t$ to the spins [\autoref{fig:clocks}(a)]. That is, it gives an energy penalty unless a state is of the form \(u_t \ket{\psi_t} \otimes \ket{t+1} + \ket{\psi_t} \otimes \ket{t}\).
        In fact, this Hamiltonian is frustration free, with a unique zero-energy ground state given by the \emph{history state}
        \begin{equation}
            \ket{\Psi_{\text{hist}}} = \frac{1}{\sqrt{T}}\sum_{t=0}^{T-1} \propagator{t}{0} \ket{\psi_0} \otimes \ket{t}, \label{eq:history_state}
        \end{equation}
        where we define the propagator \(\propagator{t_2}{t_1}\) in terms of a partial product of gates from the circuit,
        \begin{equation}
            \propagator{t_2}{t_1} = 
            \left\{
            \begin{array}{l l}
                u_{t_2-1} \cdots u_{t_1} \quad & \text{for }t_2 > t_1, \\
                \unit \quad & \text{for }t_2 = t_1, \\
                u^\dagger_{t_2} \cdots u^\dagger_{t_1-1} \quad & \text{for }t_2 < t_1.
            \end{array}
            \right.
        \end{equation}
        Note that the propagator circuits $U_{t:0}$ representing the desired quantum computation {do not appear directly} in the Hamiltonian; only their constituent local gates $\{u_t\}$ do.

        Measuring the clock in the history state Eq.~\eqref{eq:history_state} gives a uniformly random outcome $t\in \{0,..., T-1\}$. In particular it gives the outcome \(\ket{T-1}\) with probability \(1/T\). In this case, the spins are projected into a state corresponding to the result of applying the quantum computation \(\propagator{T-1}{0}\) to \(\ket{\psi_0}\),
        \begin{equation}
            [\unit \otimes \bra{T}]\ket{ \Psi_{\text{hist}}} \propto \propagator{T-1}{0} \ket{\psi_0}.
        \end{equation}
        Provided the quantum circuit size $T$ is polynomial in $n$, this outcome can be postselected efficiently. This is what allows the reduction of any problem in QMA to an instance of the local Hamiltonian problem, giving QMA-completeness~\cite{Kitaev2002book,Kempe2006localham}.

    \subsection{Periodic clocks}
        \label{subsec:PeriodicClock}

        The FK clock model is designed such that the ground state encodes the time evolution of a state under \(\propagator{t}{0}\). 
        The finite nature of this process (with an input state, a computation, and an output) translates to open boundary conditions on the clock register: states $t=0$ and $t=T-1$ are not connected by ticks.
        By changing the clock's boundary conditions to be periodic [\autoref{fig:clocks}(b)], we obtain Hamiltonians that try to energetically enforce \emph{time-periodic} history states, where the output of the computation is equal to the input. 
        That is, the eigenstates of \(H\) encode Floquet states of the circuit, interpreted as a fixed unitary evolution applied periodically in time.

        The periodic FK clock model is defined in the same spin-clock Hilbert space as in \autoref{subsec:OpenClock}, 
        but with a circuit on \(\Hspin\) with one more gate, \(\propagator{T}{0} = \prod_{t=0}^{T-1} u_t\), and an associated tick projector that replaces the initialization term,
        \begin{equation}\label{eqn:PeriodicClock}
            H = \sum_{t=0}^{T-1} \Pi_t^{\text{(tick)}}.
        \end{equation}
        The final tick projector connects states $\ket{0}$ and $\ket{T-1}$,
        \begin{multline}
            \unit - \Pi_{T-1}^\text{(tick)} =  \frac{1}{2} \big[ u_{T-1} \otimes \ket{0} - \unit \otimes \ket{T-1} \big] \\
            \times \big[ u^\dagger_{T-1} \otimes \bra{0} - \unit \otimes \bra{T-1} \big].
        \end{multline}
        Since all the unitaries $u_t$ are local, this is still a geometrically local Hamiltonian (provided one implements the clock states $\{\ket{t}\}_{t=0}^{T-1}$ in a geometrically local manner; see Sec.~\ref{subsec:MBClock} for an explicit implementation).
        
        The periodic FK clock model is typically no longer frustration free. 
        One can satisfy all tick projectors for $t = 0,..., T-2$ with a history state $\sum_{t=0}^{T-1} \ket{\psi_t}\otimes \ket{t}$ as in the open case, but the added tick projector across the boundary imposes the additional constraint that $\propagator{T}{0}\ket{\psi_0} = \ket{\psi_0}$.
        This can only be satisfied if the spectrum of $\propagator{T}{0}$ contains $+1$, which is usually not the case.
        
        Even though $H$ is frustrated, its eigenstates can be obtained exactly assuming that we are given a complete set of Floquet states \(\{\ket{\phi}\}\) for \(\propagator{T}{0}\), such that
        \begin{equation}
            \propagator{T}{0}\ket{\phi} = e^{i \phi} \ket{\phi}.
        \end{equation}
        We observe that \(H\) has Krylov sectors given by
        \begin{equation}
            \mathrm{span}\left\{ \ket{K_{t,\phi}} := \propagator{t}{0} \ket{\phi} \otimes \ket{t} : t \in \{0, ..., T-1\} \right\}.
        \end{equation}
        Indeed, noting that 
        \begin{equation}\label{eqn:PeriodicClockHop}
            H = \unit - \frac{1}{2} \sum_{t=0}^{T-1} (u_{t}\otimes  \ketbra{t+1}{t} + \text{h.c.})
        \end{equation}
        (where we identify \(\ket{T} \equiv \ket{0}\)) $H$ acts on the Krylov states $\ket{K_{t,\phi}}$ as
        \begin{equation}
            H\ket{K_{t,\phi}} = \ket{K_{t,\phi}} - \frac{1}{2} \ket{K_{t+1,\phi}} - \frac{1}{2} \ket{K_{t-1,\phi}},
        \end{equation}
        where we defined \(\ket{K_{T,\phi}} = \propagator{T}{0} \ket{\phi}\otimes \ket{0}\).
        The final tick at \(t = T-1\)  takes \(\ket{K_{T-1,\phi}}\) to \(\ket{K_{T,\phi}}\). However, by the assumption that \(\ket{\phi}\) is an eigenstate of \(\propagator{T}{0}\), we have that
        \begin{equation}
            \ket{K_{T,\phi}} = \propagator{T}{0} \ket{\phi}\otimes \ket{T} = e^{i\phi} \ket{K_{0,\phi}}.
        \end{equation}
        The final tick returns the \(\ket{K_{T-1,\phi}}\) state to the first state in the Krylov sector, but with a twisted phase. 
        Indeed, the matrix for \(H\) in the \(\ket{K_{t,\phi}}\) basis of the sector is just that for a single particle hopping on a \(T\)-site circle, with a twist of the periodic boundary conditions by \(\phi\),
        \begin{equation}
            [H]_\phi = 
            \frac{1}{2}
            \begin{pmatrix}
                2 & -1 & 0 & \cdots & -e^{i \phi} \\
                -1 & 2 & -1 & \cdots & 0 \\
                0 & -1 & 2 & \cdots & 0 \\
                \vdots & \vdots & \vdots & \ddots & \vdots \\
                -e^{-i \phi} & 0 & 0 & \cdots & 2 \\
            \end{pmatrix}.
        \end{equation}
        The eigenstates in this Krylov sector are
        \begin{equation}\label{eqn:PeriodicClockEStates}
            \ket{\Psi_{k,\phi}} = \frac{1}{\sqrt{T}}\sum_{t=0}^{T-1} e^{i (k -\phi/T)t} \ket{K_{t,\phi}}
        \end{equation}
        where \(k \in  \frac{2\pi}{T}\{0, 1, ..., T-1\}\) represents the momentum of the clock's hand. The corresponding eigenvalues are
        \begin{equation}
            E_{k, \phi} = 1- \cos(k-\phi/T).
        \end{equation}

        Some remarks about this energy spectrum are in order. First, 
        simple counting shows that the energies \(E_{k,\phi}\) account for all \(T 2^n\) eigenvalues of \(H\), so we have the entire spectrum. Notably, the spectrum is bounded: \(\mathrm{sp}(H) \subseteq [0,2]\). 
        This is despite the Hamiltonian \(H\) ostensibly being a many-body Hamiltonian. Effectively, \(H\) splits into exponentially many single-particle models, each with a bounded energy, and \(H\) inherits the spectral properties of these single-particle models.
        However, the eigenstates making up these single-particle sectors involve the infinite-temperature Floquet states \(\ket{\phi}\).

        Secondly, note that the phase $\phi$ appearing in $E_{k,\phi}$ can be modified by introducing a phase twist \(\propagator{T}{0} \mapsto e^{i \chi} \propagator{T}{0}\). 
        Such a twist can be introduced, for instance, by multiplying \(u_{T-1}\) by \(e^{i \chi}\). Inspecting \eqnref{eqn:PeriodicClockHop}, this is physically a twist of the periodic boundary conditions by \(\chi\), or equivalently a flux through the clock.
        In particular, we can set \(\chi = -\phi\) for any \(\phi\), which tunes the energy of \(\ket{\Psi_{0,\phi}}\) to \(0\). That is, by introducing a phase twist, any Floquet state can be embedded in the frustration-free ground state of the periodic FK clock model. It follows also that the FK clock Hamiltonian (assuming a nondegenerate spectrum for $U_{T-1:0}$) is frustration-free for $2^n$ distinct choices of $\chi$ in the interval $[0,2\pi)$. In this sense the Hamiltonian is typically exponentially close to frustration-free\footnote{
        Consider two distinct choices of the phase twist $\chi_1,\chi_2$. One has $\| H(\chi_1) - H(\chi_2) \|_\mathrm{op} = \| (e^{i\chi_2} - e^{i\chi_1}) \unit \otimes c^\dagger_0 c_{T-1} + \mathrm{h.c.}\|_\mathrm{op} $, since $H(\chi_1)$ and $H(\chi_2)$ differ only at the bond where the phase twist is inserted. It follows that $\| H(\chi_1) - H(\chi_2) \|_\mathrm{op} \leq 2|e^{i\chi_2}-e^{i\chi_1}| \| c_0^\dagger c_{T-1}\|_\mathrm{op} \leq 2|\chi_2 - \chi_1|$. For any given phase twist $\chi$, let $\chi^\star$ be the nearest phase twist such that $H(\chi^\star)$ is frustration-free; there are $2^n$ such choices on the unit circle. Assuming the density of states to be nearly uniform, it follows that $\| H(\chi) - H(\chi^\star)\|_\mathrm{op} \leq 2|\chi - \chi^\star| = O(2^{-n})$, so $H(\chi)$ is exponentially close in operator norm to a frustration-free Hamiltonian.
        }.

        Related to this, we also see that the finite-size gap of the periodic FK clock model is extremely small.
        For typical choices of circuit, the values of \(\phi\) will have typical level spacings of \(2\pi \cdot 2^{-n}\).
        However, the energies \(E_{k, \phi}\) have a van Hove singularity near \(E_{k,\phi} = 0\), so that the density of states is enhanced even further, and the gap above the ground state is smaller still.
        From \(E_{0,\phi} \approx (\phi/T)^2/2\), it will be \(O(4^{-n})\).

    \subsection{Geometrically local, many-handed clocks}
        \label{subsec:MBClock}

        To explore the thermal and entanglement properties of the periodic FK clock model, we must make sense of it as a geometrically local many-body Hamiltonian. 
        Hereafter, we take $T = n$ (one clock state per spin), and assume each gate $u_t$ acts on the two spins $t$, $t+1\, \mathrm{mod}\, n$, so that $\propagator{t}{0}$ becomes a one-dimensional staircase quantum circuit, sketched in \autoref{fig:clocks}. Note that the propagator $\propagator{t}{0}$ doesn't appear in the Hamiltonian, only the local gates $u_t$ do.
        In this case, the periodic FK clock model has a natural realization as a tight-binding model coupled to the spin system.

        To see this, we realize the clock Hilbert space \(\Hclock\) as the single-particle sector of a tight-binding model of fermions\footnote{
        Using fermions is not necessary.
        One could replace the terms \(c^\dagger_{t+1} c_t\) in Eq.~\eqref{eqn:MBClock} with hard-core bosonic hopping terms \(b^\dagger_{t+1} b_t\) or qubit operators \(\sigma^+_{t+1} \sigma^-_t\) [where \(\sigma_t^+ = \frac{1}{2}(X_t + i Y_t) = (\sigma_t^-)^\dagger\) and \(X_t,Y_t,Z_t\) are Pauli matrices supported on site \(t\)]. 
        The periodic FK clock model would still be realized in the single-hand subspace, where hands now refer to either bosons or qubits in the \(\ket{1}\) state.
        We choose to represent the clock in terms of fermions only because methods for solving free-fermion models are likely more familiar.}.
        Specifically, we consider the $2^n$-dimensional Fock space $\mathcal{F}_n$ of $n$ fermionic modes and decompose it into particle number sectors: $\mathcal{F}_n = \bigoplus_{M=0}^n \mathcal{H}_\mathrm{clock}^{(M)}$, where each number sector $\mathcal{H}_\mathrm{clock}^{(M)}$ has dimension $\binom{n}{M}$ and represents clock states with $M$ ``hands''. The one-handed clock Hilbert space of Sec.~\ref{subsec:OpenClock} and \ref{subsec:PeriodicClock} then is identified with $\mathcal{H}_\mathrm{clock}^{(1)}$.
        Define spinless fermion annihilation operators \(c_t\), corresponding to a fermion on site \(t\) of a length \(T\) chain. Then the Hamiltonian
        \begin{equation}\label{eqn:MBClock}
            H = \sum_{t=0}^{n-1} \unit\otimes c^\dagger_{t} c_t - \frac{1}{2}(u_{t}\otimes  c^\dagger_{t+1} c_t + \text{h.c.})
        \end{equation}
        realizes Eq.~\eqref{eqn:PeriodicClockHop} in its single-fermion sector $\mathcal{H}_\mathrm{spin} \otimes \mathcal{H}_\mathrm{clock}^{(1)}$, identifying $\ket{t}:= c^\dagger_t \ket{\Omega}$ with $\ket{\Omega}$ the fermionic vacuum state. (The Hamiltonian is zero in the fermion vacuum sector $\mathcal{H}_\mathrm{clock}^{(0)}$.)
        The dynamics in the system of spins are activated by hopping of the \emph{hands} of the clock---the fermions. This is what leads to the bounded spectrum in the single-fermion sector, even when the gates \(u_t\) may represent complicated interactions.

        Equation~\eqref{eqn:MBClock} is a local many-body Hamiltonian on a circle (\autoref{fig:clocks}).\footnote{One may generalize to other numbers of qubits or other circuit geometries by coarse-graining so that several qubits correspond to a single clock site.} 
        Indeed, it is a two-leg ladder, with a spin chain and fermion chain (the clock) parallel to each other on a circle; see also Fig.~\ref{fig:idea}. Hamiltonian terms are supported on four sites around each plaquette of the two-leg ladder: each unitary $u_t$ acts on the two qubits at sites $t$ and $t+1$ in the spin chain, and each fermionic term $c^\dagger_{t+1}c_t$ acts on modes $t$ and $t+1$ on the clock chain.

        When we additionally have that\footnote{One could demand a less restrictive condition, that $u_{n-1}$ is supported on qubit $n-1$ only. This does not change the class of valid Floquet operators \(\propagator{T}{0}\).} 
        $u_{n-1} \propto \unit$, the many-body eigenstates can still be related to Floquet states \(\ket{\phi}\) of the staircase circuit \(\propagator{n}{0}\).
        We henceforth assume that $u_{n-1} \propto \unit$.
        First, we construct the many-hand eigenstates of \(H\) as Slater determinants of single-hand states. 
        Define a length-\(M\) ordered vector of hand positions,
        \begin{multline}
            \mathbf{t} = (t_1, ..., t_M) \in \{0,...,n-1\}^M \\
            \text{such that } t_1 < t_2 < \cdots < t_M .
        \end{multline}
        Then we have Krylov sectors
        \begin{equation}
            \mathrm{span}\left\{ \ket{K_{\mathbf{t},\phi}} := \propagator{\mathbf{t}}{0} \ket{\phi} \otimes c_{\mathbf{t}}^\dagger \ket{\Omega}\right\},
        \end{equation}
        where \(\ket{\Omega}\) is the fermion vacuum state and we defined
        \begin{subequations}
        \begin{align}
            \propagator{\mathbf{t}}{0} &= \propagator{t_1}{0} \propagator{t_2}{0} \cdots \propagator{t_M}{0}, \\
            c_{\mathbf{t}}^\dagger &= c^\dagger_{t_1} c^\dagger_{t_2} \cdots c^\dagger_{t_M} .
        \end{align}
        \end{subequations}
        The states $\ket{K_{\mathbf t,\phi}}$, similarly to the single-hand case, denote a collection of clock hands at positions \(t_1, ..., t_M\), each trailed by a propagator $\propagator{t_i}{0}$ This is sketched in \autoref{fig:clocks}(c). The ordering of the product is important, as the propagators \(\propagator{t}{0}\) typically do not commute with each other. With the choice of ordering in \(\propagator{\mathbf{t}}{0}\), we find that \(H\) takes the form of a non-interacting \(M\)-particle tight-binding model with boundary conditions twisted by \(\phi\). We have that
        \begin{equation}
            H \ket{K_{\mathbf{t},\phi}} = M \ket{K_{\mathbf{t},\phi}} - \frac{1}{2} \sum_{\mathbf{t}' \in \mathrm{nn}(\mathbf{t})} \ket{K_{\mathbf{t}',\phi}},
            \label{eq:H_MBC_on_krylov}
        \end{equation}
        where \(\mathrm{nn}(\mathbf{t})\) is the set of set of hand configurations which can be obtained from \(\mathbf{t}\) with one hop.
        When $u_{n-1} \propto \unit$, applying the final gate when the \(M\)th hand completes a loop around the periodic boundary conditions trivially commutes with all the other propagators \(\propagator{t_j}{0}\). Crucially, this gives the many-hand model the correct boundary conditions:
        \begin{equation}
            \propagator{t_1}{0}\cdots \propagator{t_{M-1}}{0} \propagator{n}{0} \ket{\phi} = e^{i\phi} \propagator{0}{0} \propagator{t_1}{0} \cdots \propagator{t_{M-1}}{0} \ket{\phi}
        \end{equation}
        (recall $\propagator{0}{0} = \unit$), so that \(\mathbf{t} = (t_1, ..., t_{M-1}, n)\) is equivalent to \((0,t_1,..., t_{M-1})\), with the appropriate twist.
        
        It follows that the \(M\)-hand eigenstates of \(H\) are given by
        \begin{equation}\label{eqn:MBClockEStates}
            \ket{\Psi_{\mathbf{k},\phi}} = \sum_{\mathbf{t}} \varphi_{\mathbf{k},\phi}(\mathbf{t})\ket{K_{\mathbf{t},\phi}},
        \end{equation}
        where \(\varphi_{\mathbf{k},\phi}(\mathbf{t})\) is a Slater determinant coefficient for a state with momenta 
        $\mathbf k = (k_1,  ..., k_M)$, $k_j \in  \tfrac{2\pi}{T}\{0,1, ...,T-1\}$,
        and boundary conditions twisted by \(\phi\): 
        \begin{align}
            \varphi_{\mathbf{k},\phi}(\mathbf{t})
            & = \frac{1}{n^{M/2}} \sum_{\sigma \in S_M} \frac{\textrm{sign}(\sigma)}{M!} e^{i \sum_{j=1}^M (k_j -\phi/n) t_{\sigma(j)}},
        \end{align}
        where \(S_M\) is the set of permutations \(\sigma\) on \(M\) objects.
        Note that the only dependence of \(\varphi_{\mathbf{k},\phi}(\mathbf{t})\) on the spin system is through $\phi$.
        The energy of this eigenstate is
        \begin{equation}
            E_{\mathbf{k},\phi} = M - \sum_{k \in \mathbf{k}} \cos(k - \phi/n).
        \end{equation}
        When the number of clock hands $M$ is extensive, the bandwidth becomes extensive, as is expected of a genuine many-body Hamiltonian.
        The structure of the energy spectrum is mainly determined by the fermions in the clock, with the spin system appearing only through the momentum shift $\phi$.
        In particular, the density of states for the FK clock Hamiltonian is identical to that of a free-fermion tight-binding model in one dimension, up to a finer sampling of momenta due to the presence of $\phi$.

        We conclude this discussion with an alternative method to diagonalize the FK clock Hamiltonian, Eq.~\eqref{eqn:MBClock}, which provides a complementary perspective on the eigenstates $\ket{\psi_{\mathbf k,\phi}}$ [Eq.~\eqref{eqn:MBClockEStates}]. 
        We introduce the circuit 
        \begin{align}
            V & = (\mathsf{C}_0 \propagator{0}{0}) (\mathsf{C}_1 \propagator{1}{0}) \cdots (\mathsf{C}_{n-1} \propagator{n-1}{0}) 
            \label{eq:entangler_circuit}
        \end{align}
        where $\mathsf{C}_t \propagator{t}{0}$ represents a unitary $\propagator{t}{0}$ on the spin system controlled by the occupation of site $t$ of the clock. Since $\mathsf{C}_t \propagator{t}{0} = (\mathsf{C}_t u_{t-1})\cdots(\mathsf{C}_t u_{0})$, $V$ is a circuit made of $O(n^2)$ 3-local gates, or $O(n^3)$ geometrically local gates.
        One can verify that,
        for all $0\leq t \leq n-2$, 
        \begin{equation}
            V (\unit \otimes c^\dagger_{t+1} c_t) V^\dagger
            = u_t \otimes c^\dagger_{t+1} c_t
        \end{equation}
        and for $t = n-1$, recalling $u_{n-1}\propto \unit$,
        \begin{equation}
            V^\dagger (u_{n-1} \otimes c^\dagger_{0} c_{n-1}) V
            = \propagator{n}{0} \otimes c^\dagger_{0} c_{n-1}.
        \end{equation}
        Thus under the change of basis $V$, the FK clock Hamiltonian of Eq.~\eqref{eqn:MBClock} reads
        \begin{align}
            V^\dagger H V 
            & = \unit\otimes \left(\sum_{t=0}^{n-1} c^\dagger_t c_t - \frac{1}{2} \sum_{t=0}^{n-2} (c^\dagger_{t+1} c_t + \mathrm{h.c.}) \right) \nonumber \\ 
            & \quad -\frac{1}{2} \propagator{n}{0} \otimes c_0^\dagger c_{n-1} + \mathrm{h.c.}
        \end{align}
        Here the only operator acting nontrivially on the spin system is the Floquet unitary $\propagator{n}{0}$ (and its adjoint), which can be diagonalized separately from the clock. Projecting on a Floquet state $\ket{\phi}$ such that $\propagator{n}{0} \ket{\phi} = e^{i\phi} \ket{\phi}$ returns a tight binding Hamiltonian for the clock fermions with a twisted boundary condition $\phi$, the same as in Eq.~\eqref{eq:H_MBC_on_krylov}.
        It follows that the eigenstates $\ket{\Psi_{\mathbf k, \phi}}$ are obtained as 
        \begin{align}
            \ket{\Psi_{\mathbf k, \phi}}
            & = V (\ket{\phi} \otimes c^\dagger_{\mathbf k, \phi}\ket{\Omega}), \label{eq:MBClock_dressed_ES}
        \end{align}
        as can be also verified by direct calculation. 
        This form illustrates more explicitly the structure of the eigenstates [Eq.~\eqref{eqn:MBClockEStates}].
        It also shows that, given an eigenstate $\ket{\Psi_{\mathbf k,\phi}}$, it is possible to efficiently prepare the corresponding Floquet state $\ket{\phi}$ of the spin system by acting with $V^\dagger$ (a polynomial-depth circuit) and discarding the clock. 

    \subsection{Volume-law eigenstate entanglement entropy}
        \label{subsec:VolumeLawEEE}

        In this section, we show that the eigenstates of the periodic FK clock Hamiltonian are volume-law entangled whenever the eigenstates of the input quantum circuit obey the eigenstate thermalization hypothesis (ETH)~\cite{D'Alessio2016ETHreview,Deutsch2018ETHreview}.

        As we have seen, the eigenstates of the periodic FK clock Hamiltonian \(H\) [\eqnref{eqn:MBClock}] are expressed in terms of eigenstates of a unitary circuit $U_{n:0}$. 
        The vast majority of these Floquet eigenstates are generically expected to obey the ETH at infinite temperature:
        \begin{definition}[Floquet ETH] \label{def:eth}
            Let $\{\ket{i}\}$ be eigenstates of a Floquet operator on \(n\) qubits without conservation laws and let $O$ be a local operator with operator norm \(\norm{O}_{\mathrm{op}}\leq 1\). The states $\{\ket{i}\}$ are said to obey the \emph{eigenstate thermalization hypothesis} (ETH) at infinite temperature, or \emph{Floquet ETH}, if
            \begin{equation}
                \bra{i} O \ket{j} = \delta_{ij} \frac{{\rm Tr}(O)}{2^n} + 2^{-n/2} R_{ij}^O, \label{eq:floquet_eth}
            \end{equation}
            where $|R_{ij}^O|$ is of order 1. 
            If \eqnref{eq:floquet_eth} holds with \(i=j\) and all local operators \(O\), the state \(\ket{i}\) is said to obey \emph{diagonal ETH} at infinite temperature, or \emph{diagonal Floquet ETH}.
        \end{definition}
        The statement of ETH usually demands that the factors \(R_{ij}^O\) be erratically varying or random-like in some way which is rarely quantified precisely. They cannot be completely uncorrelated in actual local models~\cite{Pappalardi2022freeeth,Pappalardi2025fulleth,Shi2023localeigenstate,Hahn2024floqueteth}. We will only demand that they are \(O(1)\) in system size when the support of the operator \(O\) is fixed.
        Also note that Definition~\ref{def:eth} gives no condition on the eigenvalues \(e^{i \phi}\) of the Floquet unitary, only on the eigenstates \(\ket{i}\). To relate ETH to equilibration (smallness of fluctuations in expectation values about their steady state value), additional non-resonance conditions on such eigenvalues are typically required~\cite{Srednicki1999approach,Linden2009equilibration,Pilatowskycameo2025thermalization,Huang2024randomproduct}. 
        We separate these non-resonance conditions from the condition on eigenstates, and refer only to the latter as ETH.
        
        A consequence of diagonal Floquet ETH is that each Floquet eigenstate is highly entangled. One can show that local reduced density matrices are extremely close to maximally mixed:
        \begin{lemma}\label{lem:ETHImpliesCloseInTraceNorm}
            Suppose \(\ket{i}\) obeys diagonal Floquet ETH and that for all operators \(O\) supported in a local subsystem \(A\) of size \(\ell\) with \(\norm{O}_{\mathrm{op}}\leq 1\) we have \(|R^O_{ii}| \leq C_\ell\). Then
            \begin{subequations}
            \begin{align}
                \norm{\Tr_{\bar{A}}\ketbra{i}{i} - 2^{-\ell} \unit_A}_{\Tr} &\leq  C_\ell 2^{-n/2}, \label{eqn:TrNormBound} \\
                \norm{\Tr_{\bar{A}}\ketbra{i}{i}}_{\mathrm{op}} &\leq 2^{-\ell} + C_\ell 2^{-n/2},\label{eqn:OpNormBound}
            \end{align}
            \end{subequations}
            where \(\norm{\cdot}_{\Tr}\) is the trace norm and \(\norm{\cdot}_{\mathrm{op}}\) is the operator norm.
        \end{lemma}
        \begin{proof}
            For \eqnref{eqn:TrNormBound}, use \(\norm{\Delta}_{\Tr} = \sup_{\norm{O}_{\mathrm{op}}\leq1} |\Tr(O \Delta)|\) and directly apply Definition~\ref{def:eth}.
            Then \eqnref{eqn:OpNormBound} follows from the triangle inequality and \(\norm{\Delta}_{\mathrm{op}} \leq \norm{\Delta}_{\Tr}\).
        \end{proof}
        This implies a lower bound on entanglement entropy:
        for any local subsystem $A$, we have $S_A(\ket{i}) \geq |A|\ln(2) + O(2^{-n/2})$.
        It is natural to expect that this volume-law scaling of entanglement entropy is inherited by the Hamiltonian eigenstates \(\ket{\Psi_{\mathbf{k},\phi}}\) [\eqnref{eqn:MBClockEStates}], as we show is indeed the case. Note that in \autoref{sec:LFSRModel} we introduce a model where Floquet ETH can be proven rigorously for almost all eigenstates, without any physical assumptions. All the results of this section hold for that choice of \(\{u_t\}_t\). 
        
        Here we focus on the case of a one-handed clock, while the many-handed case (conceptually analogous but technically more involved) is discussed in Appendix~\ref{app:clock}. 
        We assume diagonal Floquet ETH not just for the Floquet state $\ket{\phi}$, but also for the partially time-evolved states $\propagator{t}{0} \ket{\phi}$ for all $t<n$; note that these are eigenstates of \(\propagator{t}{0}\propagator{n}{t}\), the same periodic staircase circuit with a different choice of starting time for the period. Thus, if \(\ket{\phi}\) obeys diagonal Floquet ETH, all of \(\{\propagator{t}{0}\ket{\phi}\}_{t}\) typically will as well.
        
        We focus on contiguous subintervals $A$ of the spin-clock ladder: $A = A_s \cup A_c$ where $A_s$ and $A_c$ describe the same interval $\{t_1,t_1+1,\dots t_1+\ell-1 \text{ mod } n\}$ in the spin and clock chains respectively, as sketched in Fig.~\ref{fig:idea}. These are also the relevant contiguous subsystems when viewing the spin-clock ladder as a single chain of 4-level qudits (each qudit comprising one rung of the ladder: one qubit and one fermionic mode).

        \begin{thm}[Volume-law entanglement from Floquet ETH] \label{thm:volumelaw_1hand}
            If \(\propagator{t}{0}\ket{\phi}\) obeys diagonal Floquet ETH (Definition~\ref{def:eth}) for all \(t < n\), then all eigenstates \(\ket{\Psi_{k, \phi}}\) of the one-handed periodic FK clock Hamiltonian [\eqnref{eqn:PeriodicClockHop}] which are built from \(\ket{\phi}\) are volume-law entangled: 
            the entanglement entropy of any interval $A$ of length $\ell \leq n/2$ obeys
            \begin{equation}
                S_A(\ket{\Psi_{k,\phi}}) \geq \ell \ln(2)
            \end{equation} 
            provided that
            \begin{equation}
                \frac{\ell \ln(2)}{n/2} \geq C_\ell 2^{-n/2} + C_\ell^2 2^{\ell-n}, \label{eq:ell_criterion}
            \end{equation}
            where \(C_\ell\) is as in Lemma~\ref{lem:ETHImpliesCloseInTraceNorm}.
            Eq.~\eqref{eq:ell_criterion} holds in particular for any constant $\ell$ at sufficiently large $n$.
        \end{thm}
        
        \begin{proof}
            We lower-bound the entanglement entropy by the second R\'enyi entropy,
            \begin{equation}
                S_A(\ket{\Psi_{k,\phi}}) \geq \SecRenyi{A}{\ket{\Psi_{k,\phi}}} := -\ln \Tr(\rho_A^2),
            \end{equation}
            where $\rho_A$ is the reduced density matrix of $\ket{\Psi_{k,\phi}}$:
            \begin{align}
                \rho_A 
                & = \sum_{t,t'=0}^{n-1} \varphi_{k,\phi}(t) \varphi^\ast_{k,\phi}(t') 
                \sigma_{A_s}^{(t,t')} \otimes {\rm Tr}_{\bar{A}_c}(\ketbra{t}{t'}), 
            \end{align}
            with $A_s$ and $A_c$ the subsystems of the spin and clock chains, $\varphi_{k,\phi}(t) = n^{-1/2} e^{i(k-\phi/n)t}$ are the free fermion eigenfunctions, and
            \begin{align}
                \sigma_{A_s}^{(t,t')}
                & = {\rm Tr}_{\bar{A}_s}(U_{t:0}\ketbra{\phi} U_{t':0}^\dagger).
            \end{align}
            For the operator on the clock Hilbert space,
            we have ${\rm Tr}_{\bar{A}_c}(\ketbra{t}{t'}) = \ketbra{t}{t'}_{A_c}$ if both $t,t'$ belong to $A_c$, otherwise ${\rm Tr}_{\bar{A}_c}(\ketbra{t}{t'}) = \delta_{t,t'} \ketbra{\Omega}_{A_c}$, with $\ket{\Omega}$ the vacuum state. This gives
            \begin{align}
                \rho_A
                & = \frac{1}{n} \sum_{t,t'\in A} e^{i(k-\phi/n)(t-t')} \sigma_{A_s}^{(t,t')} \otimes \ketbra{t}{t'}_{A_c} \nonumber \\
                & \quad + \frac{1}{n} \sum_{t\notin A}  \sigma_{A_s}^{(t,t)} \otimes \ketbra{\Omega}_{A_c}.
            \end{align}
            The purity of $\rho_A$ is calculated straightforwardly using orthogonality of the clock states ($\braket{t}{t'} = \delta_{t,t'}$, $\braket{\Omega}{t} = 0$), giving
            \begin{align}\label{eqn:PurityFrobenius}
                {\rm Tr}(\rho_A^2) 
                & = \frac{1}{n^2} \sum_{t,t'\in A} \left\| \sigma^{(t,t')}_{A_s}\right\|_F^2
                + \frac{1}{n^2} \left\| \sum_{t\notin A} \sigma^{(t,t)}_{A_s} \right\|_F^2 
            \end{align}
            with $\|O\|_F = \sqrt{{\rm Tr}(O^\dagger O)}$ the Frobenius norm. Here we used the fact that $\sigma^{(t,t')} = (\sigma^{(t',t)})^\dagger$.

            Let us focus on the ``diagonal'' terms $\sigma^{(t,t)}$ first. Note that these are simply the reduced density matrices of the states $U_{t:0}\ket{\phi}$, which we assume obey diagonal Floquet ETH (Definition~\ref{def:eth}).
            Applying Lemma~\ref{lem:ETHImpliesCloseInTraceNorm} to \(\propagator{t}{0}\ket{\phi}\), we obtain
            \begin{align}
                {\rm Tr}[\sigma_{A_s}^{(t,t)} \sigma_{A_s}^{(t',t')}] &= {\rm Tr}[\sigma_{A_s}^{(t,t)} (\sigma_{A_s}^{(t',t')} - 2^{-\ell} \unit_{A_s})] + 2^{-\ell} \nonumber \\
                & \leq 2^{-\ell} + \|\sigma_{A_s}^{(t,t)}\|_{\mathrm{op}} \| \sigma^{(t',t')}_{A_s} - 2^{-\ell} \unit_{A_s}\|_{\Tr} \nonumber \\ 
                &\leq 2^{-\ell} +  C_\ell 2^{-\ell -n/2} + C_\ell^2 2^{-n}
            \end{align}
            for some constant \(C_\ell\) as in Lemma~\ref{lem:ETHImpliesCloseInTraceNorm}. We also used the triangle inequality and \(|\Tr(O_1 O_2)| \leq \norm{O_1}_{\mathrm{op}} \norm{O_2}_{\Tr}\).

            For the ``off-diagonal'' terms $\sigma^{(t,t')}_{A_s}$ with $t \neq t'$, we must distinguish two cases, depending on whether or not $A$ contains the origin, $t=0$. 
            \begin{itemize}
                \item[(i)] $0 \notin A$. Using the composition property of propagators $U_{t:0} = U_{t:t'} U_{t':0}$, we can write write $\sigma^{(t,t')}_{A_s} = \sigma^{(t,t)}_{A_c} U_{t:t'} $, noting that $U_{t:t'}$ is entirely supported in $A_s$ and thus can be taken outside the partial trace. 
                \item[(ii)] $0 \in A$. In this case $t$ and $t'$ may lie on opposite sides of the origin, making $U_{t:t'}$ a nonlocal string that spans $\bar{A}_s$. 
                To avoid this issue we use the composition property in a different way. Assuming $t > t'$ without loss of generality, we have $U_{t:0}\ket{\phi} = U_{t:n} U_{n:0} \ket{\phi} = e^{i\phi} U_{n:t}^\dagger \ket{\phi}$, using the fact that $\ket{\phi}$ is a Floquet eigenstate. Thus $\sigma_{A_s}^{(t,t')} = e^{i\phi} U_{n:t}^\dagger \sigma_{A_s}^{(0,0)} U_{t':0}^\dagger$.
            \end{itemize}
            In both cases we managed to reduce $\sigma^{(t,t')}_{A_s}$ to a product of a diagonal term, 
            $ \sigma^{(t,t)}_{A_s} \approx 2^{-\ell } \unit_{A_s}$, 
            and unitary operators supported inside $A_s$. It follows that
            \begin{equation}
                \| \sigma^{(t,t')}_{A_s}\|_F^2 = \Tr[\sigma^{(t,t)}_{A_s} \sigma^{(t,t)}_{A_s}] \leq 2^{-\ell} +  C_\ell 2^{-\ell -n/2} + C_\ell^2 2^{-n},
            \end{equation}
            as before.
            
            Adding all terms in \eqnref{eqn:PurityFrobenius}, we get
            \begin{equation}
                {\rm Tr}(\rho_A^2) \leq \frac{\ell^2 + (n - \ell)^2}{n^2} \left[2^{-\ell} +  C_\ell 2^{-\ell -n/2} + C_\ell^2 2^{-n} \right], \label{eq:purity_factors}
            \end{equation}
            and thus
            \begin{subequations}
            \begin{align} 
                S_A(\ket{\Psi_{k,\phi}}) &= - \ln {\rm Tr}(\rho_A^2) \\
                &\geq \ell \ln(2) + \ln\left[ \frac{n^2}{\ell^2 + (n-\ell)^2} \right] \nonumber \\
                &\qquad- \ln\left[1+ C_\ell 2^{-n/2} + C_\ell^2 2^{\ell -n} \right], \\
                &\geq \ell \ln(2) + \frac{\ell \ln(2)}{n/2} - C_\ell 2^{-n/2} - C_\ell^2 2^{\ell-n}, \label{eqn:SABound}
            \end{align}
            \end{subequations}
            where we used \(x \geq \ln(1+x)\), and used \(\ell \leq n/2\) so that \(\ln[n^2/(\ell^2+(n-\ell)^2)] \geq 2\ln(2) \ell/n\) (which follows from concavity of the logarithm).
            When the conditions of the theorem hold, the combination of the latter three terms in \eqnref{eqn:SABound} is nonnegative, and we simply have \(S_A(\ket{\Psi_{k,\phi}}) \geq \ell \ln(2)\).
        \end{proof}

        Theorem~\ref{thm:volumelaw_1hand} applies to all eigenstates $\ket{\Psi_{k,\phi}}$, and thus in particular to the ground state $\ket{\Psi_{0,\phi}}$, provided the associated Floquet state $\ket{\phi}$ meets the theorem's assumptions. If such an eigenstate exists anywhere in the spectrum of $U_{n:0}$, then it can always be embedded in the ground state of $H$ by introducing a phase twist $\chi = -\phi$ in the clock, e.g., by changing $u_{n-1} = c \unit$ to $u_{n-1} = c e^{i\chi} \unit$.
        Thus  Theorem~\ref{thm:volumelaw_1hand} shows that the ground state of the periodic FK clock model can be volume-law entangled.
        Even further, the entanglement-entropy-density of the ground state is nearly the maximum possible in the single clock hand sector: the dimension of the single hand sector is \(n 2^n\), corresponding to a maximal entropy density of \(\ln 2 + \tfrac{\ln n}{n}\), and the entanglement-entropy-density of $\ket{\Psi_{0,\phi}}$ is \(\ln 2\).

        We note that Eq.~\eqref{eq:purity_factors} has a physically intuitive structure: the purity factors into two contributions, $2^{-\ell}$ (up to a small error term) from the spin chain, and $(\ell/n)^2 + (1-\ell/n)^2$ from the clock. The latter comes from the probability, $\ell/n$, of finding the clock hand in region $A_c$ in an eigenfunction $\varphi_{k,\phi}(t)$. This factorization in fact holds more generally. In Appendix~\ref{app:clock} we prove a generalization of Theorem~\ref{thm:volumelaw_1hand} to the many-handed case:

        \theoremstyle{plain}
        \newtheorem{innercustomgeneric}{\customgenericname}
        \providecommand{\customgenericname}{}
        \newcommand{\newcustomtheorem}[2]{%
          \newenvironment{#1}[1]
          {%
           \renewcommand\customgenericname{#2}%
           \renewcommand\theinnercustomgeneric{##1}%
           \innercustomgeneric
          }
          {\endinnercustomgeneric}
        }
        \newcustomtheorem{customthm}{Theorem}
        \begin{customthm}{A.1}
            Let \(A\) be any subinterval of the spin-clock system.
            Then the second R\'eyni entropy of any many-hand clock eigenstate \(\ket{\Psi_{\mathbf{k},\phi}}\) is bounded below as
            \begin{equation}
            \SecRenyi{A}{\ket{\Psi_{\mathbf{k},\phi}}} \geq \SecRenyi{A}{c^\dagger_{\mathbf{k},\phi}\ket{\Omega}} + \min_{\mathbf{t}} \SecRenyi{A}{\propagator{\mathbf{t}}{0} \ket{\phi}},
            \end{equation}
            where the minimum is over hand configurations with no more hands than \(\ket{\Psi_{\mathbf{k},\phi}}\), and
            \begin{equation}
            c^\dagger_{\mathbf{k},\phi} = \sum_{\mathbf{t}} \varphi_{\mathbf{k},\phi}(\mathbf{t}) c_{\mathbf{t}}^\dagger
            = \prod_{k \in \mathbf{k}} \left[\frac{1}{\sqrt{T}}\sum_{t=0}^{T-1} e^{i (k -\phi/T)t} c_t^\dagger\right]
            \end{equation}
            is the creation operator for a Slater determinant state.
        \end{customthm}

        Provided all $\{\propagator{\mathbf t}{0} \ket{\phi} \}_{\mathbf t}$ are volume-law entangled, then this theorem generalizes our result of volume-law entanglement to many-handed clock eigenstates as well. We note again that the contributions to the entanglement lower bound from the clock and the spin system add up. In particular the clock contributes the entanglement of a free-fermion Slater determinant state of momenta $\mathbf k$, which can be itself volume-law (in highly excited states)~\cite{Lai2015fermionentanglement} or logarithmic (in the Fermi sea ground state)~\cite{Wolf2006fermionlogentropy,Gioev2006freefermionentropy}.

    \subsection{Infinite-temperature spin expectation values}
    \label{subsec:SpinExpectations}
    
    We have noted informally above that the FK clock ground state's entanglement entropy carries a contribution of roughly $\ell \log(2)$ from the spin chain in addition to the clock's contribution. This would imply that the spin chain by itself (i.e., tracing out the clock chain) should be near-maximally entangled, thus at infinite temperature. 
    Here we make this statement precise. 

    Let us consider a many-handed clock eigenfunction,
    \begin{equation}
        \ket{\Psi_{\mathbf k, \phi}}
        = \sum_{\mathbf t} \varphi_{\mathbf k, \phi}(\mathbf t) (\propagator{\mathbf t}{0} \ket{\phi}) \otimes c^\dagger_{\mathbf t}\ket{\Omega}.
    \end{equation}
    Its reduced density matrix on the spin system is 
    \begin{align}
        \rho_\mathrm{spin} & 
        = \mathrm{Tr}_\mathrm{clock} (\ket{\Psi_{\mathbf k, \phi}}\bra{\Psi_{\mathbf k, \phi}}) \nonumber \\ 
        & = \sum_{\mathbf t} |\varphi_{\mathbf k, \phi}(\mathbf t)|^2 \propagator{\mathbf t}{0} \ket{\phi} \bra{\phi} \propagator{\mathbf t}{0}^\dagger. 
    \end{align}
    Lastly, the reduced density matrix on a local subsystem $A_s$ of the spin system is 
    \begin{equation}
        \rho_{A_s} 
        = \sum_{\mathbf t} |\varphi_{\mathbf k, \phi}(\mathbf t)|^2 \mathrm{Tr}_{\bar{A}_s} (\propagator{\mathbf t}{0} \ket{\phi} \bra{\phi} \propagator{\mathbf t}{0}^\dagger).
    \end{equation}
    Assuming Floquet ETH for all the partially time-evolved Floquet states $\propagator{\mathbf t}{0}$, each term in the sum is equal to the maximally mixed state, up to exponentially small (in system size $n$) error, so $\rho_{A_s} \approx \unit_{A_s} / 2^{|A_s|}$. 
    It follows that any local operator supported exclusively on the spin chain sees a maximally mixed density matrix, and thus has an expectation value consistent with infinite temperature. 

    It is worth highlighting that not all operators have infinite-temperature behavior in general; for example, the ``energy density'' terms $u_t \otimes c_{t+1}^\dagger c_t + \mathrm{h.c.}$ are by definition sensitive to the state's energy.
    This is an unavoidable feature of local Hamiltonian ground states. 
    However, in FK clock eigenstates, only operators that act nontrivially on the clock chain can exhibit energy dependence. Operators supported only on the spins invariably show infinite-temperature behavior, including in the ground state. Combined with the volume-law scaling of entanglement, this behavior of observables substantiates the unusual coexistence of ``infinite-temperature'' behavior alongside more conventional local ground state behavior in FK clock eigenstates, as sketched in Fig.~\ref{fig:idea}.
    
\section{Floquet circuits with provably thermal eigenstates}
    \label{sec:LFSRModel}

    So far, our results rely on the assumption of ETH for the input Floquet circuit. While the ETH is widely believed to hold for generic choices of circuit~\cite{D'Alessio2016ETHreview,Liao2022fieldtheoryapproacheigenstate}, it is desirable to have models where our results are provable. In fact, there are few (if any) local circuit or Hamiltonian models which provably obey the ETH. Beyond our application, it is broadly beneficial to the study of ergodicity and chaos in quantum systems to have an exactly solvable model of eigenstate thermalization.
    
    To this end, we introduce a class of Floquet circuits whose eigenstates are provably thermal. These circuits are based on linear feedback shift registers (LFSRs)~\cite{Tausworthe1965random,Massey_1969_shiftregister,Klein2013lfsr}, classical circuits that play a role in pseudorandom number generators. Below, we introduce LFSRs and their implementation as staircase quantum circuits, derive their eigenstates (\autoref{subsec:lfsr}),  prove they obey infinite-temperature ETH (\autoref{subsec:lfsr_eth}), and use this to prove volume-law entanglement in the corresponding FK clock Hamiltonian ground state (\autoref{subsec:volumelaw_lfsr}).

    While our results are phrased in terms of a staircase circuit, this is not essential to the construction. The same eigenstate thermalization properties can be shown to hold for a brickwork circuit constructed from the same gates.\footnote{
    The staircase and brickwork versions of the circuit are related by a Clifford similarity transformation, so they share the same spectrum and the same set of Pauli matrix elements. The latter is all we use to conclude eigenstate thermalization.
    }

    \subsection{LFSR circuits}
        \label{subsec:lfsr}

    An LFSR is a reversible update rule for a classical state of $n$ bits, $(z_0, ...,  z_{n-1}) \in \{0,1\}^n$, given by\footnote{For convenience we use a backward shift, instead of the more standard forward shift, to obtain a quantum staircase circuit with the same layout used in \autoref{sec:FKClockModel}. }
    \begin{equation}\label{eq:lfsr}
    z_i' = \left\{ 
    \begin{array}{ll}
        z_{i+1} \quad  & \text{for }i<n-1, \\
        z_0 + \sum_{i=1}^{n-1} a_i z_i \bmod 2 \quad & \text{for }i=n-1.
    \end{array}
    \right.
    \end{equation}
    The system size $n$ and binary coefficients $a_i$ completely specify the model. In words, the LFSR shifts all bits by one place, then the \emph{output} bit ($i = 0$) is fed back into the \emph{input} position ($i = n-1$) along with a linear function of the other bits.
    
    As a reversible finite-state machine, each LFSR partitions the state space into a set of orbits. An LFSR is called \emph{maximal} if it splits the state space into just two orbits: the all-zero state (which always forms a trivial orbit by itself, due to linearity) and all $2^n-1$ non-zero states. Maximal LFSRs are known to exist for all values of $n$, can be constructed efficiently, and are in one-to-one correspondence with primitive polynomials of degree $n$ over the binary field $\mathbb{F}_{2}$; see Appendix~\ref{app:LFSRReview} for more details.  

    The LFSR update rule, \eqnref{eq:lfsr}, is naturally implemented as a staircase circuit of \(\mathsf{SWAP}\) and \(\mathsf{CNOT}\) gates, 
    \begin{multline}\label{eq:uf_lfsr}
        U_{n:0} = (\mathsf{CNOT}_{n-2,n-1})^{a_{n-1}} \mathsf{SWAP}_{n-2,n-1} \times \cdots \\
        \cdots \times (\mathsf{CNOT}_{0,1})^{a_{1}} \mathsf{SWAP}_{0,1} , 
    \end{multline}
    corresponding to $u_i = (\mathsf{CNOT}_{i,i+1})^{a_i} \mathsf{SWAP}_{i,i+1}$ for $0\leq i \leq n-2$ and $u_{n-1} = \unit$. Here 
    $\mathsf{CNOT}_{i,j}$ is the controlled NOT gate with control bit $i$ and target bit $j$ and the $a_i$ are the binary coefficients from Eq.~\eqref{eq:lfsr}. This is sketched in \autoref{fig:lfsr}(a). 
    The staircase circuit structure makes these models well-suited to our FK clock Hamiltonian construction. 
    
    Equation~\eqref{eq:uf_lfsr} belongs to {several} special classes of quantum circuits:
    reversible classical circuits (also known as ``automaton'' or ``permutation'' circuits~\cite{Gopalakrishnan2018operatorgrowth,Gopalakrishnan2018automata,Iaconis2021automata,Bertini2025permutation}), which do not generate coherence; 
    Clifford circuits, which do not generate ``magic'' (non-stabilizerness)~\cite{Aaronson2004simulationstabilizer,Veitch2014resourcetheory}; 
    dual-unitary circuits~\cite{Bertini2025exactlysolvablemanybodydynamics}, which are unitary along both the time and space directions. 
    Despite these special features, $U_{n:0}$ has eigenstates that behave similarly to random states, as we show next. 

    Let us consider a maximal LFSR with update rule $\mathbf{z}' = \alpha \mathbf{z}$, with $\alpha$ a $n\times n$ binary matrix implementing \eqnref{eq:lfsr}. 
    The Floquet unitary in Eq.~\eqref{eq:uf_lfsr} acts on the compuational basis as $U_{n:0} \ket{\mathbf z} = \ket{\alpha \mathbf z}$. From this, it is clear that the all-zero state $\ket{\boldsymbol{0}}$ is an eigenstate with eigenvalue $+1$. The remaining $2^n-1$ eigenstates span the nontrivial orbit and are given by
    \begin{equation}
        \ket{\psi_q} = \frac{1}{\sqrt{2^n-1}} \sum_{j=0}^{2^n-2} \omega^{qj} \ket{\alpha^j \boldsymbol{1}},
        \label{eq:lfsr_eig_alpha}
    \end{equation}
    where $\boldsymbol{1} = (00\dots 01)$, $\omega = e^{2\pi i /(2^n-1)}$, and $q \in \{0,1, ..., 2^n-2\}$. 
    The eigenstate $\ket{\psi_q}$ has eigenvalue $\omega^{-q}$. 
    For $q=0$, \eqnref{eq:lfsr_eig_alpha} gives a nearly disentangled state, $\ket{\psi_0} \propto \ket{+}^{\otimes n} - 2^{-n/2} \ket{0}^{\otimes n} $, with eigenvalue $+1$. So the $+1$ eigenspace of $U_{n:0}$ is spanned by two product states, $\ket{0}^{\otimes n}$ and $\ket{+}^{\otimes n}$. All other eigenspaces, $\omega^q$ with $q\neq 0$, are non-degenerate. 
    Hereafter we refer to all $q\neq 0$ eigenstates in Eq.~\eqref{eq:lfsr_eig_alpha} as the {\it nontrivial eigenstates} of the maximal LFSR. 

    It is useful at this point to introduce the discrete logarithm on the finite field $\mathbb{F}_{2^n}$, defined as follows: for all bitstrings $\mathbf{z} \neq \boldsymbol{0}$, $\log_\alpha(\mathbf{z})$ is the unique element $j \in \mathbb{Z}_{2^n-1}$ such that $\alpha^j \boldsymbol{1} = \mathbf{z}$. With this definition, the LFSR eigenstates \eqnref{eq:lfsr_eig_alpha} may be written as 
    \begin{equation}
    \ket{\psi_q} = \frac{1}{\sqrt{2^n-1}} \sum_{\mathbf z \neq \boldsymbol{0}} \omega^{q\log_\alpha(\mathbf z)} \ket{\mathbf z}.
        \label{eq:lfsr_eig_log}
    \end{equation}
    These states were introduced in Refs.~\cite{VanDam2002gausssums,VanDam2003discretelogs} under the name of ``chi states'' in the context of quantum algorithms for algebraic problems, and also recently appeared in Ref.~\cite{Kim2025catalyticzrotationsconstanttdepth} in the context of magic gate catalysis for quantum computing.
    
    The discrete logarithm $\log_\alpha(\mathbf{z})$ is known to behave ``pseudorandomly'', without obvious patterns or correlations across inputs.\footnote{Computing $\log_\alpha(\mathbf{z})$ is also thought to be a hard problem for classical computers, where quantum computers may have exponential advantage~\cite{Shor1994factoring,Shor1997factoring,VanDam2003discretelogs}.}
    As a consequence, the phases $\{\omega^{q\log_\alpha(\mathbf z)}:\mathbf z \in \{0,1\}^n\setminus \{\boldsymbol{0}\}\}$ vary erratically between different basis states. This suggests a similarity between the nontrivial ($q\neq 0$) eigenstates of Eq.~\eqref{eq:lfsr_eig_alpha} and pseudorandom phase states~\cite{Ji2018pseudrandomstates,Brakerski2019pseudorandom}, which in turn behave similarly to Haar-random states. We make this intuition precise in the following.

    \begin{figure*}
        \centering
        \includegraphics[width=0.99\textwidth]{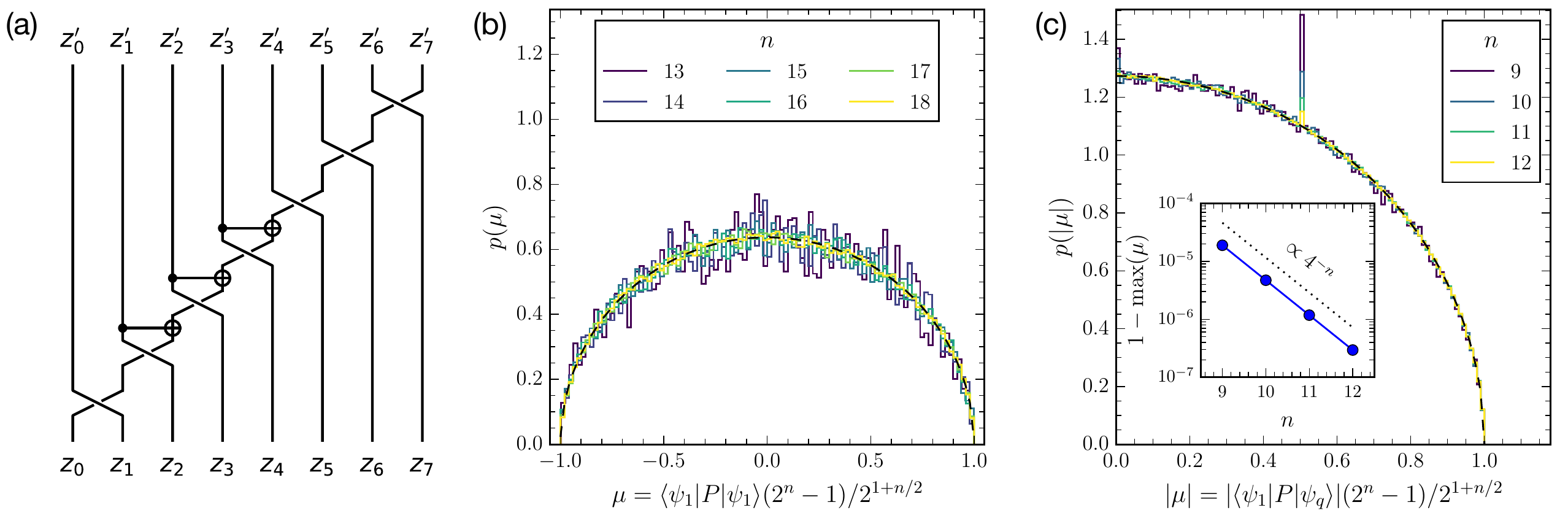}
        \caption{(a) Circuit diagram for a LFSR with $n = 8$ qubits, implementing the update rule Eq.~\eqref{eq:lfsr} with $a_2 = a_3 = a_4 = 1$, $a_i = 0$ otherwise. 
        (b-c) Numerical tests of Lemma~\ref{lemma:pauli}. 
        (b) Distribution of Pauli expectation values on the nontrivial eigenstate $\ket{\psi_1}$ of an $n$-bit maximal LFSR, normalized by the value of the upper bound in Lemma~\ref{lemma:pauli}, for various system sizes. All $4^n-1$ non-idenity Pauli operators are included. The dashed line represents a semicircle distribution, $p(\mu) = (2/\pi) \sqrt{1-\mu^2}$ for $|\mu|\leq 1$. All computed values respect the bound ($|\mu|\leq 1$). 
        (c) Distribution of Pauli matrix elements between nontrivial eigenstates $\ket{\psi_1}$ and $\ket{\psi_q}$ of an $n$-bit maximal LFSR, normalized as in (b). All $4^n-1$ nontrivial Pauli operators and all $2^n-2$ values of $q\neq 0$ are included. The inset shows that the matrix elements come exponentially close to saturating the bound, $\max |\mu|\approx 1-\mathrm{const.}\times 4^{-n}$.
        There is a visible spike in \(p(|\mu|)\) above the semicircle background at \(|\mu|=1/2\); this arises from Pauli operators that are either purely \(X\)-type or \(Z\)-type, which are a fraction $\sim 2^{1-n}$ of the total (Appendix~\ref{app:PauliDetails}).}
        \label{fig:lfsr}
    \end{figure*}
    
    \subsection{Eigenstate thermalization}
        \label{subsec:lfsr_eth}
    
    We begin by showing that all matrix elements of non-identity Pauli operators between nontrivial LFSR eigenstates are exponentially small. 
    We then use this fact to establish a precise statement of eigenstate thermalization, from which near-maximal entanglement entropy follows. 

    \begin{lemma}[Bound on Pauli matrix elements] \label{lemma:pauli}
        For any maximal $n$-bit LFSR and any Pauli operator $P \neq \unit$, the eigenstates in \eqnref{eq:lfsr_eig_log} obey 
        \begin{equation}
            |\bra{\psi_q} P \ket{\psi_{q'}}| \leq \frac{2^{1+n/2}}{2^n-1}
        \end{equation}
        for all $q,q'\neq 0$. 
    \end{lemma}
    We remark that this holds for {\it all} non-identity Pauli operators, not just local ones. 
    The proof of Lemma~\ref{lemma:pauli}, reported in Appendix~\ref{app:proof}, uses tools from algebraic number theory (bounds on ``mixed-character sums''~\cite{Weil_1948_exponentialsums,castro_2000_mixedsums}) that formalize the ``pseudorandom'' nature of the discrete logarithm.
    
    We test Lemma~\ref{lemma:pauli} numerically by computing Pauli matrix elements between nontrivial eigenstates, shown in \autoref{fig:lfsr}(b-c). We find excellent agreement; the normalized matrix elements
    \begin{equation}
        \mu := \frac{2^n-1}{2^{1+n/2}} \bra{\psi_q} P \ket{\psi_{q'}}
    \end{equation}
    are empirically found to follow an approximate semicircle distribution, $p(|\mu|) \propto \sqrt{1-|\mu|^2}$, very nearly saturating the bound $|\mu|\leq 1$. 

    Lemma~\ref{lemma:pauli} also implies a precise statement of eigenstate thermalization for maximal LFSRs.
    \begin{thm}[LFSR eigenstate thermalization] \label{thm:eth}
        For any maximal $n$-bit LFSR, all nontrivial eigenstates obey the infinite-temperature ETH:
        \begin{equation}
            \bra{\psi_q} O \ket{\psi_{q'}} = 2^{-n} {\rm Tr}(O) \delta_{q,q'} + 2^{-n/2} R^O_{qq'}, \label{eq:eth_lfsr}
        \end{equation}
        for all $q,q'\neq 0$ and all local operators $O$, where $|R^O_{qq'}| \leq 2^{2+\ell} \| O \|_{\rm op}$ and $\ell$ is the size of the support of $O$. 
        It follows that nontrivial eigenstates have entanglement entropy
        \begin{equation}
            S_A(\ket{\psi_q}) \geq \ell \ln(2) - 2^{4+2\ell -n}, \label{eq:entanglement_bound}
        \end{equation}
        which is near-maximal for $\ell = o(n)$.
    \end{thm}
    \begin{proof}
        First let us consider the case of a Pauli operator, $O = P$. Lemma~\ref{lemma:pauli} immediately gives Eq.~\eqref{eq:eth_lfsr} with $|R^P_{qq'}| \leq 4$. 
        For a general operator $O$, let us decompose in the Pauli basis: $O = \sum_{P\in A} o_P P$, with $o_P = {\rm Tr}(PO) / 2^n$, with $A$ the support of $O$ and $P\in A$ denoting Pauli operators supported in $A$. 
        Plugging into the left hand side of Eq.~\eqref{eq:eth_lfsr} yields the correct $\delta_{q,q'}$ term, as well as $R^O_{qq'} = 2^{n/2} \sum_{P\in A}' o_P \bra{\psi_q} P \ket{\psi_{q'}} $ with the primed sum denoting $P \neq \unit$. 
        The triangle inequality and Lemma~\ref{lemma:pauli} give
        \begin{align}
            \left|R^O_{qq'}\right| 
            & \leq 2^{n/2} \sum_{P\in A} {}^{'} |o_P| |\bra{\psi_q} P \ket{\psi_{q'}}|
            \leq 4 \sum_{P\in A} |o_P|.
        \end{align}
        Then Cauchy-Schwarz gives 
        \begin{equation}
            \sum_{P\in A} |o_P| \leq 2^{\ell } \left( \sum_{P\in A} |o_P|^2 \right)^{1/2} = 2^{\ell -n/2} \| O \|_F
        \end{equation}
        where $\|O\|_F = [{\rm Tr}(O^\dagger O)]^{1/2}$ 
        is the Frobenius norm. The norm inequality $\| O \|_F \leq 2^{n/2} \| O \|_{\rm op}$ then gives
        \begin{equation}
            \left|R^O_{qq'}\right| \leq 2^{2+\ell} \| O \|_{\rm op},
        \end{equation}
        which concludes the proof of Eq.~\eqref{eq:eth_lfsr}.

        Next we prove Eq.~\eqref{eq:entanglement_bound}. Volume-law entanglement is a standard consequence of ETH, since Eq.~\eqref{eq:eth_lfsr} implies that local reduced density matrices must be very close to the fully mixed state (Lemma~\ref{lem:ETHImpliesCloseInTraceNorm}). In this case we can use the slightly more powerful result of Lemma~\ref{lemma:pauli}.
        Since the von Neumann entropy is lower-bounded by the second R\'enyi entropy, we have
        \begin{align}
            S_A(\ket{\psi_q}) 
            & \geq -\ln\left[ \frac{1}{2^{\ell}}\sum_{P\in A} \bra{\psi_q} P \ket{\psi_q}^2 \right],
        \end{align}
        where $P\in A$ denotes Pauli operators supported in region $A$, and we have used the fact that ${\rm Tr}(\rho_A^2) = 2^{-\ell} \sum_{P\in A} [{\rm Tr}(\rho_A P)]^2$ for any density matrix $\rho_A$ on $A$.
        We then invoke Lemma~\ref{lemma:pauli} with $q = q'$ to bound $|\bra{\psi_q} P \ket{\psi_{q}}|^2 \leq 2^{4-n}$ for \(P \neq \unit\). 
        Splitting the Pauli sum into $P = \unit $ and $P\neq \unit $ ($4^{\ell}-1$ terms) gives
        \begin{align}
            S_A(\ket{\psi_q}) 
            & \geq \ell \ln(2) - \ln\left[ 1 + \frac{4^\ell-1}{2^{n-4}} \right] \nonumber \\
            & \geq \ell \ln(2) - 2^{4-n} (4^\ell -1),
        \end{align}
        where we also used $\ln(1+x) \leq x$. \eqnref{eq:entanglement_bound} follows immediately.
    \end{proof}
    We note that Theorem~\ref{thm:eth} does not require the subsystem $A$ to be geometrically local. The result depends only on the volume of $A$, i.e., the number of qubits on which $O$ acts. 

    \subsection{Volume-law entangled ground states}
        \label{subsec:volumelaw_lfsr}

    In \autoref{subsec:VolumeLawEEE} we proved that periodic FK clock Hamiltonians have volume-law entangled ground states (Theorem~\ref{thm:volumelaw_1hand}), conditional on the validity of Floquet ETH (Definition~\ref{def:eth}), or at least volume-law entanglement of the states \(\propagator{\mathbf{t}}{0}\ket{\phi}\) (Theorem~\ref{thm:ManyHandEEE}).
    Here we show that, when the gate sequence $\{u_t\}_t$ implements a maximal LFSR, the same result can be proven {\it without any physical assumptions}:
    \begin{thm}[Volume-law entangled ground states] \label{thm:volumelaw_gs}
        For any maximal $n$-bit LFSR, any nontrivial number of hands \(0<M<n\), and any flux $\chi$ at least $\pi/(2^n-1)$ away from 0 (if \(M\) is odd) or \(\pi\) (if \(M\) is even) modulo $2\pi$, the ground state $\ket{\Psi}$ of the corresponding periodic FK clock Hamiltonian is volume-law entangled: 
        \begin{equation}
            S_A(\ket{\Psi}) \geq \ell \ln(2)
        \end{equation}
        for any interval $A$ of length \(\ell \leq n/2-6\).
    \end{thm}
    \begin{proof}
        For odd \(M\), the Krylov sectors of the trivial ``scar'' LFSR eigenstates \(\ket{0}^{\otimes n}\) and \(\ket{+}^{\otimes n}\) contain the ground state if and only if the flux \(\chi\) is in \([-\pi/(2^n-1),\pi/(2^n-1)]\) modulo \(2\pi\). Otherwise, the ground state is found in a Krylov sector associated to a nontrivial (ETH-obeying) LFSR eigenstate $\ket{\psi_q}$, $q\neq 0$.
        For even \(M\), the lowest energy Slater determinant state is achieved with flux \(\chi=\pi\), so that the ground state is similarly non-trivial provided \(\chi\not\in \pi + [-\pi/(2^n-1),\pi/(2^n-1)]\) modulo \(2\pi\).

        Then we can apply Theorem~\ref{thm:ManyHandEEE}. Indeed, the states \(\propagator{\mathbf{t}}{0} \ket{\psi_q}\) appearing in that theorem all differ from \(\ket{\psi_q}\) only by applying a Clifford circuit. Thus, Lemma~\ref{lemma:pauli} still applies to \(\propagator{\mathbf{t}}{0} \ket{\psi_q}\): for all Pauli $P\neq \unit$, $\bra{\psi_q}U_{\mathbf t:0}^\dagger P U_{\mathbf t:0} \ket{\psi_q} = \bra{\psi_q}P'\ket{\psi_q}$ where $P'\neq \unit$ is another Pauli. All these expectation values are exponentially small due to Lemma~\ref{lemma:pauli}. Thus the lower bound on entanglement, Eq.~\eqref{eq:entanglement_bound}, still applies to $U_{\mathbf t:0} \ket{\psi_q}$. 
        Inspecting the proof of Theorem~\ref{thm:eth}, this bound actually applies to the second R\'eyni entropy, and we have
        \begin{equation}
            \min_{\mathbf{t}}\SecRenyi{A}{\propagator{\mathbf{t}}{0} \ket{\psi_q}} \geq \ell \ln(2) - 2^{4+2 \ell - n}.
        \end{equation}
        The clock hand contribution is bounded as (Lemma~\ref{lem:FermionEntropyBound})
        \begin{equation}
            \SecRenyi{A}{c_{\mathbf{k},\phi}^\dagger \ket{\Omega}} \geq \frac{\ell^2}{n^2},
        \end{equation}
        so we have for the ground state (Theorem~\ref{thm:ManyHandEEE})
        \begin{equation}\label{eqn:SAWithError}
            S_A(\ket{\Psi}) \geq \SecRenyi{A}{\ket{\Psi}} \geq \ell \ln(2) + \frac{\ell^2}{n^2} - 2^{4+2 \ell - n}.
        \end{equation}
        
        First note that \(S_A(\ket{\Psi}) = 0 \geq 0 \cdot \ln(2)\) trivially holds for \(\ell = 0\).
        Some elementary manipulation shows that the latter two terms in \eqnref{eqn:SAWithError} always give a nonnegative contribution when \(1 \leq \ell \leq n/2-6\). The contribution is positive precisely when
        \begin{equation}\label{eqn:PosCondition}
            \frac{\ell^2}{n^2} \geq 2^{4+2 \ell - n}
            \iff 
            \frac{\ell}{2^\ell} \geq 8 \frac{n/2}{2^{n/2}}.
        \end{equation}
        The function \(\ell 2^{-\ell}\) is a nonincreasing function for integers \(\ell \geq 1\), so the inequality holds for all \(1 \leq \ell \leq n/2 - d\) if
        \begin{equation}
            \frac{n/2 - d}{2^{n/2-d}} \geq 8\frac{n/2}{2^{n/2}}
            \iff
            1 - \frac{d}{n/2} \geq 8 \cdot 2^{-d}.
        \end{equation}
        We further bound \(1-2d/n\) from below as \(1-d/(1+d)\) (again using \(n/2 \geq 1+d\)), arriving at a sufficient condition which only involves \(d\),
        \begin{equation}
            1- \frac{d}{1+d} \geq 8 \cdot 2^{-d}.
        \end{equation}
        This holds for \(d \geq 6\).
    \end{proof}

    Theorem~\ref{thm:volumelaw_gs} is another main result of this work: a new rigorous construction of local Hamiltonians in one dimension with volume-law entangled ground states, qualitatively different from prior examples based on rainbow states~\cite{Gottesman2010entanglementvsgap,Irani2010groundstateentanglement,Vitagliano2010volumelaw,Ramirez2014volumelaw} or Motzkin and Fredkin spin chains~\cite{Bravyi2012criticality,Movassagh2016supercritical,Zhang2017extensive,Salberger2018fredkin}. 
    Notably, the volume-law scaling of entanglement entropy is established for all contiguous intervals up to the half cut.
    Additionally, while we used a loose lower bound for the entropy of the clock (valid for arbitrary free-fermion eigenstates), in the ground state with \(M \sim \nu n\) (\(0<\nu<1\)) the fermions form a logarithmically entangled Fermi sea state. This gives $S_A \approx \ell \ln(2) + (1/3) \ln(\ell)$ up to further subleading corrections~\cite{Wolf2006fermionlogentropy,Gioev2006freefermionentropy}, satisfying the lower bound $S_A \geq \ell \ln (2)$ up to and including the half-cut $\ell = n/2$. 

    We remark that with an encoding of the clock register with hard-core bosons or qubits, rather than fermions, the even-odd effect in Theorem~\ref{thm:volumelaw_gs} would be absent\footnote{We acknowledge private communication from Hank Green for pointing out this even-odd $M$ effect.}.
    The pathological value of the flux that selects one of the trivial LSFR eigenstates as the ground state would always be at \(\chi = 0\).
    
\section{Discussion}
    \label{sec:Discussion}

    We have introduced two families of solvable models: a local many-body variant of the Feynman-Kitaev clock (FK clock) with periodic boundary conditions, and a class of Floquet circuits based on linear feedback shift registers (LFSRs).
    The former provides a way of embedding Floquet states throughout the spectrum of a static local Hamiltonian, while the latter provides an analytically tractable instance of ETH-obeying many-body Floquet states.
    Using the LFSR circuits in the periodic FK clock construction, we obtain a one-dimensional Hamiltonian where all but a vanishing fraction of eigenstates are provably volume law entangled.
    The volume law scaling of entanglement can be realized by the ground state, and applies to arbitrary intervals of up to one half the global system size.

    Both classes of models we introduced are valuable as exactly-solvable counterexamples or edge-cases for conventional lore regarding the behavior of entanglement in many-body systems.
    
    First, the periodic FK clock Hamiltonians give rise to a qualitatively new type of area-law violation in local one-dimensional models.
    While other examples of volume-law entangled ground states in one dimension are known~\cite{Gottesman2010entanglementvsgap,Irani2010groundstateentanglement,Bravyi2012criticality,Movassagh2016supercritical,Zhang2017extensive,Salberger2018fredkin}, they display a special ``rainbow-like'' entanglement structure~\cite{Alexander2019rainbowtn}, such that not all subsystems actually obey the volume law.
    The eigenstates of the FK clock model, being descendants of Floquet states, have an entanglement structure much more similar to that of infinite-temperature eigenstates, where all subsystems obey the volume law.
    
    With that said, the FK clock eigenstates do not obey ETH at any temperature.
    They retain some of the observable physics of the free fermion clock hands~\cite{Lydzba2023quadratic}, despite the presence of nontrivial entanglement between these degrees of freedom and the spin system.
    For instance, eigenstates of the same energy but opposite net hand momentum can be distinguished by measuring local ``dressed'' current operators, 
    $i(c_{t+1}^\dagger c_t \otimes u_t - \textrm{h.c.})$, as these act like ordinary current operators on a free-fermion state [see Eq.~\eqref{eq:MBClock_dressed_ES}].
    Meanwhile, observables supported only on the spin system appear essentially featureless in FK clock eigenstates.
    
    Having such behavior in the ground state is unconventional, even in a gapless system, and it carries nontrivial implications for
    the generic structure of ground states beyond their entanglement properties.
    As a consequence for condensed matter physics: if one believes that low-energy and long-wavelength degrees of freedom of local Hamiltonians in one-dimension are described by quantum field theory, then the field theory for the FK clock model must simultaneously account for a zero-temperature free-fermion piece (infinite imaginary-time radius) and a featureless infinite-temperature piece (zero imaginary-time radius),  with nontrivial entanglement between them. Alternatively, the 
    FK clock models could fail to flow to a well-defined field theory, in violation of common assumptions. Both outcomes would have interesting implications for many-body physics, particularly if the FK clock Hamiltonians can represent critical points between distinct gapped phases of matter.
    As a consequence for quantum information theory: we see from \eqnref{eq:MBClock_dressed_ES} that the periodic FK clock ground state is related to a tensor product of a free-fermion state and a Floquet eigenstate by a circuit of polynomial (in system size $L$) depth. 
    This implies that results on the complexity of Floquet states (either the circuit complexity of their preparation on quantum hardware, or the computational complexity of their classical representation) can be ported to static local Hamiltonian ground states. Specifically, if there is any family of Floquet eigenstates with super-polynomial circuit complexity (that is, the number of gates required to prepare the eigenstates grows faster than any polynomial in \(L\)), then it follows that Hamiltonian ground states can also have super-polynomial circuit complexity.
    The same goes for the complexity of classical representation. Ground states of local Hamiltonians in one dimension often admit an efficient representation in terms of matrix product states. If one could prove the existence of a Floquet state without an efficient classical representation (as is expected based on random matrix theory ideas), the FK clock construction would turn that into a statement about local Hamiltonian ground states.

    The utility of the LFSR circuits as solvable case studies for ETH and chaos
    is likely even further reaching.
    While the ETH is well-established numerically~\cite{Rigol2008thermalization,D'Alessio2016ETHreview,Deutsch2018ETHreview}, and has been understood through heuristics for decades~\cite{Deutsch1991statmechclosed,Srednicki1994chaosthermalization}, we are not aware of any other local models which have been rigorously shown to satisfy it.
    At best, a variety of models have been shown to satisfy \emph{weak ETH}, where eigenbasis matrix elements vanish slower than exponentially in system size~\cite{Biroli2010rarefluctation,Mori2016weakeigenstatethermalizationlarge,Lydzba2024normalweaketh,Vikram2025bypassingeigenstatethermalizationexperimentally}.
    The LFSR circuits thus provide a unique handle on understanding ETH and its implications as we continue to move from a qualitative understanding of ergodicity and chaos in quantum systems to a quantitative one.
    
    Indeed, LFSR circuits meet many characterizations of quantum chaos and thermalization, but drastically fail others.
    Beginning with thermalization: as we have proved, they satisfy ETH in the sense of Definition~\ref{def:eth}~\cite{D'Alessio2016ETHreview,Deutsch2018ETHreview}.
    Nonetheless, the statistics of their quasienergies are nothing like a random matrix~\cite{Potters2020rmtbook}---they are evenly spaced, maximally violating the ``no-resonance condition'' that is often used to show equilibration~\cite{Srednicki1999approach,Linden2009equilibration,Pilatowskycameo2025thermalization,Huang2024randomproduct}.
    As a consequence, despite satisfying ETH and being ergodic---in the sense that time-averaged expectation values of local observables reproduce infinite temperature expectation values---LFSR dynamics need not equilibrate~\cite{Srednicki1999approach}. As LFSR circuits map \(Z\)-basis states to \(Z\)-basis states, the expectation value of a local Pauli \(Z\) operator only ever jumps between \(\pm 1\) when applying the circuit, and never approaches the equilibrium value of \(0\).
    Another consequence of the uniform level spacing is that the Heisenberg times and Poincar\'e recurrence times coincide for LFSR circuits, whereas generically the latter is exponentially larger~\cite{Bocchieri1957recurrence,Ozmaniec2024recurrence}.
    The relation of LFSR circuits to quantum chaos seems even more contentious.
    Conforming with a notion of quantum chaos, they have exactly maximal Krylov complexity~\cite{Nandy2025krylov} up to order \(2^n-1\) when acting on either $X$ or $Z$ basis product states (except \(\ket{0}^{\otimes n}\) or \(\ket{+}^{\otimes n}\)) or Pauli operators (except the identity).\footnote{This makes them concrete examples of what Ref.~\cite{Vikram2023cyclicergodicity} called ``quantum cyclic ergodicity''. They also exhibit ``quantum cyclic aperiodicity'', but over too short a time scale for there to be nontrivial consequences on the spectrum, as studied in Ref.~\cite{Vikram2023cyclicergodicity}.}
    However, they are at the intersection of permutation circuits and Clifford circuits, generating no coherence and no magic, and thus no state or operator entanglement when acting on physically relevant classes of input states or operators~\cite{Prosen2007opentangle,Prosen2007simulations}.
    Altogether, the LFSR circuits represent a useful edge case to disentangle and discriminate between inequivalent~\cite{Harrow2023thermalizationnoeth}---but often conflated---ideas in the study of chaos and thermalization.
    It would be enlightening to make a more comprehensive assessment of how they may or may not be considered chaotic~\cite{Pandey2020agp}.

    Much of the phenomenology we have discussed should be highly unstable to generic perturbations.
    Indeed, the single-hand sector of the FK clock model has a bandwidth which is independent of system size. 
    Typical extensive perturbations will result in a bandwidth which is extensive, presumably completely rearranging the eigenstates in the thermodynamic limit.
    Introducing local impurities to the model will maintain the \(O(1)\) bandwidth, but may still drastically change ground state properties, especially as the density of states is very large near the ground state.
    The many-hand sectors already have extensive bandwidth, but this does not guarantee that volume law entanglement in the ground state will persist.
    Generally, we expect typical extensive perturbations will open a gap in the FK clock model. (Similar instability was numerically observed in the Fredkin chain~\cite{Adhikari2019deformingfredkin}.)
    Perturbations to the LFSR circuits should make them more in line with conventional expectations for chaotic quantum circuits, but it is not clear that any analytic control can be maintained as such a perturbation is introduced.
    
    Despite this fragility to generic perturbations, the class of periodic FK clock models is still quite large. Perturbations of the Hamiltonian which maintain the form of \eqnref{eqn:MBClock}, only changing the \(u_t\) gates, are expected to typically maintain the infinite-temperature properties of the FK clock eigenstates, including volume law entanglement.
    This gives a manifold of Hamiltonians with extensive dimension which are all expected to have volume law entangled ground states. While still fine-tuned, this is a much larger family than those obtained by prior constructions~\cite{Bravyi2012criticality,Movassagh2016supercritical,Zhang2017extensive,Salberger2018fredkin}.
    
    Indeed, the precise form of \eqnref{eqn:MBClock} can even be slightly relaxed.
    A simple extension of the FK clock model which does not spoil the volume law entanglement of eigenstates is to give the \(c^\dagger_{t} c_t\) and \(c^\dagger_{t+1} c_t\) terms in \eqnref{eqn:MBClock} \(t\)-dependent amplitudes. For instance, adding arbitrarily weak random disorder will localize the clock hands~\cite{Anderson1958localization}, but nonetheless retain volume law entanglement in the system's eigenstates.
    Even including density-density interactions \(c^\dagger_t c_t c^\dagger_{t'} c_{t'}\) only introduces interactions among the dressed fermions [\eqnref{eq:MBClock_dressed_ES}], and so preserves the decomposition into Krylov sectors (though it makes the sectors interacting, spoiling solvability).
    
    More significant extensions to the FK clock model may include variants in higher dimensions~\cite{Gosset2015spacetimecircuit,Zhang2023coupledfredkin,Balasubramanian20232dentanglement}, or with other input circuits. 
    Indeed, there are many important classes of models which can be realized as Floquet circuits, including discrete time crystals~\cite{Khemani2019briefhistorytimecrystals,Else2020timecrystals} and topological Floquet phases~\cite{Harper2020floquettopologyreview,Long2024beyondeigenstate}.
    These Floquet circuits generate intrinsically driven phenomena which are impossible to realize with static, local Hamiltonians, so the embedding of their eigenstates in the clock model must greatly affect their properties.
    How, then, might the nontrivial Floquet dynamics of these circuits appear in either the dynamics or spectrum of the periodic FK clock model? Can these exhibit any remnants or signatures of intrinsically driven physics?
    We leave these questions for future work.

    Finally, we comment on another possible implication of our embedding of Floquet states in static Hamiltonian spectra, relating to the preparation of quasi-stationary states of Floquet circuits.
    There are well-developed computational techniques for finding ground (or low-energy) states of local Hamiltonians, particularly in one dimension, including variational methods based on tensor networks~\cite{White1992dmrg,Catarina2023dmrg} or neural quantum states~\cite{Lange2024neural}. Targeting Floquet states tends to be more challenging.
    Our construction of periodic FK clock Hamiltonians presents an interesting opportunity to map the latter problem into the former and thus take advantage of ground-state numerical methods.
    While it will typically be very difficult to find the exact ground state, due to the gap in the periodic FK clock Hamiltonian being exponentially small ($\propto 4^{-n}$), even a superposition of low-energy eigenstates is sufficient for some purposes.
    By applying the circuit in \eqnref{eq:MBClock_dressed_ES}, a low-energy state of the spin-clock ladder can then be efficiently converted to a state of the spin system which is quasi-stationary under Floquet dynamics.
    Provided that the variational energy of the spin-clock state is much smaller than the finite-size energy gap of the free-fermion sector ($\propto n^{-2}$), the resulting spin state has a narrow quasi-energy distribution and is correspondingly long-lived under the Floquet dynamics.
    An interesting direction for future work will be to compare the efficiency, accuracy, and achievable system sizes in this approach against existing techniques~\cite{Khemani2016dmrgx,Zhang2017dmrgx,Yu2017shiftinvert,Seirant2020polyfed,Luitz2021floquetpolyfed}.

\begin{acknowledgements}
    We thank Anushya Chandran, Pieter Claeys, Sam Garratt, Sarang Gopalakrishnan, Jeongwan Haah, and Vedika Khemani for discussions. We also thank Shankar Balasubramanian and Amit Vikram for feedback on the manuscript.
    MI acknowledges the hospitality of the Department of Physics at Stanford University during the beginning of this work.
    This research was supported in part by grant NSF PHY-2309135 to the Kavli Institute for Theoretical Physics (KITP).
    DML was supported by a Stanford Q-FARM Bloch Fellowship, the US Department of Energy, Office of Science (Award No.\ DE-SC0019380), and a Packard Fellowship in Science and Engineering (PI: Vedika Khemani).
\end{acknowledgements}

\appendix

\section{Entanglement of many-handed clock eigenstates \label{app:clock}}
    
    In this appendix, we prove a many-hand generalization of the volume-law entanglement property of \autoref{subsec:VolumeLawEEE}.

    \begin{thm}\label{thm:ManyHandEEE}
        Let \(A\) be any subinterval of the spin-clock system.
        Then the second R\'eyni entropy of any many-hand clock eigenstate \(\ket{\Psi_{\mathbf{k},\phi}}\) is bounded below as
        \begin{equation}\label{eqn:ReyniBound}
        \SecRenyi{A}{\ket{\Psi_{\mathbf{k},\phi}}} \geq \SecRenyi{A}{c^\dagger_{\mathbf{k},\phi}\ket{\Omega}} + \min_{\mathbf{t}} \SecRenyi{A}{\propagator{\mathbf{t}}{0} \ket{\phi}},
        \end{equation}
        where the minimum is over hand configurations with no more hands than \(\ket{\Psi_{\mathbf{k},\phi}}\), and
        \begin{equation}
        c^\dagger_{\mathbf{k},\phi} = \sum_{\mathbf{t}} \varphi_{\mathbf{k},\phi}(\mathbf{t}) c_{\mathbf{t}}^\dagger
        = \prod_{k \in \mathbf{k}} \left[\frac{1}{\sqrt{T}}\sum_{t=0}^{T-1} e^{i (k -\phi/T)t} c_t^\dagger\right]
        \end{equation}
        is the creation operator for a Slater determinant state.
    \end{thm}

    The second R\'eyni entropy is a lower bound for the von Neumann entanglement entropy, so Theorem~\ref{thm:ManyHandEEE} also gives a bound for the entanglement entropy.
    
    In particular, if all of the states \(\propagator{\mathbf{t}}{0} \ket{\phi}\) obey a volume law in the second R\'eyni entropy, then \(S_A({\ket{\Psi_{\mathbf{k},\phi}}}) \geq \min \SecRenyi{A}{\propagator{\mathbf{t}}{0} \ket{\phi}} \geq c |A|\) also obeys a volume law. 
    The interpretation of \(\propagator{\mathbf{t}}{0} \ket{\phi}\) is roughly that it is built from eigenstates for the time-shifted circuit \(\propagator{t}{0} \propagator{n}{t}\), so it is natural that if one of these \(\propagator{\mathbf{t}}{0} \ket{\phi}\) states is volume law entangled for a fixed \(\ket{\phi}\), then they all are.
    More concretely, such volume law entanglement holds for the LFSR circuits of \autoref{sec:LFSRModel}, where applying any Clifford circuit such as \(\propagator{\mathbf{t}}{0}\) to an eigenstate gives a state which still obeys the bound of Lemma~\ref{lemma:pauli}, and thus has the ETH properties that descend from that lemma.

    To prove Theorem~\ref{thm:ManyHandEEE}, we first establish a simple lemma.

    \begin{lemma}\label{lem:VWDecomp}
        Given any length \(M\) vectors of hand positions \(\mathbf{t}\) and \(\mathbf{t}'\), we have
        \begin{equation}
        \propagator{\mathbf{t}}{0} = V \propagator{\mathbf{s}}{0},
        \quad\text{and}\quad
        \propagator{\mathbf{t}'}{0} =  W \propagator{\mathbf{s}}{0},
        \end{equation}
        where \(\mathbf{s} = \max( \mathbf{t}, \mathbf{t}')\) is the pointwise maximum of \(\mathbf{t}\) and \(\mathbf{t}'\), and
        \begin{equation}
            V = \prod_{j \,:\, t_j < t'_j} \propagator{t_j}{t'_j}, \quad
            W = \prod_{j \,:\, t_j > t'_j} \propagator{t'_j}{t_j},
        \end{equation}
        where the ordering of the product is from small \(j\) (on the right) to large \(j\) (on the left).
    \end{lemma}
    \begin{proof}
        We essentially just have to check that the stated decomposition works. 
        It is convenient to phrase this as an induction on \(M\).
        For \(M=0\), we have \(\propagator{\mathbf{t}}{0} = \unit\), so the induction base certainly holds. Assuming that the lemma holds for \(M\) clock hands, we consider vectors of \(M+1\) clock hands \((t_0, \mathbf{t}) = (t_0, t_1, ..., t_M)\) and \((t'_0, \mathbf{t}')\). 
        Then we have
        \begin{equation}
            \propagator{(t_0, \mathbf{t})}{0} = \propagator{t_0}{0} \propagator{\mathbf{t}}{0},
        \end{equation}
        and similarly with primes.
        Thus, from the induction hypothesis,
        \begin{equation}
            \propagator{(t_0, \mathbf{t})}{0} = \propagator{t_0}{0} V \propagator{\mathbf{s}}{0},
            \quad\text{and}\quad
            \propagator{(t'_0,\mathbf{t}')}{0} = \propagator{t'_0}{0} W \propagator{\mathbf{s}}{0}.
        \end{equation}

        We observe from their definition that \(V\) is supported at most on \([t_1,t_M']\), and similarly that \(W\) is supported on \([t'_1, t_M]\). The propagators \(\propagator{t_0^{(\prime)}}{0}\) are supported on \([0,t_0^{(\prime)}]\), and to be a valid hand configuration we must have \(t^{(\prime)}_0 < t^{(\prime)}_1\). Thus, \(V\) has disjoint support from \(\propagator{t_0}{0}\), and \(W\) has disjoint support from \(\propagator{0}{t_0'}\). This gives
        \begin{equation}
            \propagator{(t_0, \mathbf{t})}{0} =  V \propagator{t_0}{0} \propagator{\mathbf{s}}{0},
            \quad\text{and}\quad
            \propagator{(t'_0,\mathbf{t}')}{0} =  W \propagator{t'_0}{0} \propagator{\mathbf{s}}{0}.
        \end{equation}

        Suppose that \(t_0 < t_0'\). Then we write
        \begin{equation}
            \propagator{t_0}{0} \propagator{\mathbf{s}}{0} = \propagator{t_0}{t'_0} \propagator{t'_0}{0} \propagator{\mathbf{s}}{0} = \propagator{t_0}{t'_0} \propagator{(s_0,\mathbf{s})}{0},
        \end{equation}
        where \((s_0, \mathbf{s})\) is the pointwise maximum of \((t_0, \mathbf{t})\) and \((t'_0, \mathbf{t}')\). Thus,
        \begin{equation}
            \propagator{(t_0, \mathbf{t})}{0} =  V'  \propagator{(s_0,\mathbf{s})}{0}
        \end{equation}
        where \(V'= V \propagator{t_0}{t'_0}\) is the new \(V\). A similar calculation holds for the new \(W\) when \(t_0 > t_0'\). The lemma follows by induction.
    \end{proof}

    We now proceed to Theorem~\ref{thm:ManyHandEEE}.

    \begin{proof}[Proof of Theorem~\ref{thm:ManyHandEEE}]
    We need to compute the purity of the reduced density matrix on the subsystem \(A\):
    \begin{equation}
    \rho^A_{\mathbf{k}, \phi} = \sum_{\mathbf{t}, \mathbf{t}'} \varphi_{\mathbf{k},\phi}(\mathbf{t}) \varphi^*_{\mathbf{k},\phi}(\mathbf{t}') \Tr_B \ketbra{K_{\mathbf{t},\phi}}{K_{\mathbf{t}',\phi}},
    \end{equation}
    where \(B\) is the complement of \(A\).
    Using that \(\ketbra{K_{\mathbf{t},\phi}}{K_{\mathbf{t}',\phi}}\) is a tensor product of operators on \(\Hspin\) and \(\Hclock\), we have
    \begin{multline}
    \Tr_B \ketbra{K_{\mathbf{t},\phi}}{K_{\mathbf{t}',\phi}}
    = \Tr_{B}[\propagator{\mathbf{t}}{0} \ketbra{\phi}{\phi} \propagator{0}{\mathbf{t}'}] \\
    \otimes \Tr_{B}[ c^\dagger_\mathbf{t} \ketbra{\Omega}{\Omega} c_{\mathbf{t}'}].
    \end{multline}
    Write \(\mathbf{t}_{A}\) for the ordered vector of hand positions consisting of elements of \(\mathbf{t}\) which lie inside \(A\).
    Then, denoting \(\ket{\mathbf{t}} = c^\dagger_\mathbf{t} \ket{\Omega}\), the partial trace of the clock state is simply
    \begin{equation}
    \Tr_{B} \ketbra{\mathbf{t}}{\mathbf{t}'} = \delta_{\mathbf{t}_{B} \mathbf{t}'_{B}} \ketbra{\mathbf{t}_{A}}{\mathbf{t}'_{A}},
    \end{equation}
    where \(\ket{\mathbf{t}_{A}}\) is a state on the subsystem.
    This gives us an expression for the purity
    \begin{widetext}
    \begin{multline}\label{eqn:ClockPurityFull}
    \Tr \left[(\rho^A_{\mathbf{k}, \phi})^2\right] = \sum_{\mathbf{t}_1, \mathbf{t}'_1, \mathbf{t}_2, \mathbf{t}'_2} \varphi_{\mathbf{k},\phi}(\mathbf{t_1}) \varphi^*_{\mathbf{k},\phi}(\mathbf{t}'_1) \varphi_{\mathbf{k},\phi}(\mathbf{t}'_2) \varphi^*_{\mathbf{k},\phi}(\mathbf{t}_2) \delta_{\mathbf{t}_{1,B} \mathbf{t}'_{1,B}} \delta_{\mathbf{t}_{2,B} \mathbf{t}'_{2,B}} \delta_{\mathbf{t}_{1,A} \mathbf{t}_{2,A}} \delta_{\mathbf{t}'_{1,A} \mathbf{t}'_{2,A}} \\
    \times \Tr \left[\Tr_{B}[\propagator{\mathbf{t}_1}{0} \ketbra{\phi}{\phi} \propagator{0}{\mathbf{t}'_1}]
    \Tr_{B}[\propagator{\mathbf{t}'_2}{0} \ketbra{\phi}{\phi} \propagator{0}{\mathbf{t}_2}] \right].
    \end{multline}
    \end{widetext}
    
    The key parts of \eqnref{eqn:ClockPurityFull} are the constraints from the \(\delta\)-functions and the trace involving \(\ket{\phi}\). Equation~\eqref{eqn:ReyniBound} follows from showing that
    \begin{multline}\label{eqn:FactorBound}
    \Tr \left[\Tr_{B}[\propagator{\mathbf{t}_1}{0} \ketbra{\phi}{\phi} \propagator{0}{\mathbf{t}'_1}]
    \Tr_{B}[\propagator{\mathbf{t}'_2}{0} \ketbra{\phi}{\phi} \propagator{0}{\mathbf{t}_2}] \right] \\
    = \exp\left[-\SecRenyi{A}{\propagator{\mathbf{s}}{0} \ket{\phi}}\right]
    \end{multline}
    for some \(\mathbf{s}\) whenever the combination of \(\delta\)-functions is nonzero.
    Indeed, bounding each of these terms in \eqnref{eqn:ClockPurityFull} by taking the minimum entropy over \(\mathbf{s}\), we factor them out of the sum and recognize the remaining expression as being the purity of \(c^\dagger_{\mathbf{k},\phi} \ket{\Omega}\),
    \begin{equation}
    \Tr \left[(\rho^A_{\mathbf{k}, \phi})^2\right] \leq \exp\left[- \SecRenyi{A}{c^\dagger_{\mathbf{k},\phi} \ket{\Omega}} - \min_{\mathbf{s}} \SecRenyi{A}{\propagator{\mathbf{s}}{0} \ket{\phi}} \right].
    \end{equation}
    Taking the logarithm gives the lower bound for the second R\'eyni entropy stated in \eqnref{eqn:ReyniBound}.
    
    In proving \eqnref{eqn:FactorBound}, we distinguish two cases:
    \begin{enumerate}
        \item\label{itm:NoLastBond} the subsystem \(A\) does not contain the sites \(\{0,T-1\}\), and
        \item\label{itm:YesLastBond} \(A\) does contain \(\{0,T-1\}\).
    \end{enumerate}
    In case~\ref{itm:NoLastBond}, the subsystem \(A\) does not straddle the final bond completing the periodic boundary conditions, and any \(\mathbf{t}\) is decomposed as \(\mathbf{t} = (\mathbf{t}_{B1}, \mathbf{t}_A, \mathbf{t}_{B2})\), where the subscript denotes which subsystem the hands are contained in. (Some of the components may be empty.)
    In case~\ref{itm:YesLastBond} we instead have \(\mathbf{t} = (\mathbf{t}_{A1}, \mathbf{t}_B, \mathbf{t}_{A2})\).
    
    First consider case~\ref{itm:NoLastBond}, and apply Lemma~\ref{lem:VWDecomp} to \eqnref{eqn:FactorBound}, giving
    \begin{align}
        \propagator{\mathbf{t}_1}{0} \ketbra{\phi}{\phi} \propagator{0}{\mathbf{t}'_1} &= V_1 \propagator{\mathbf{s}_1}{0} \ketbra{\phi}{\phi} \propagator{0}{\mathbf{s}_1} W_1^\dagger, \\
        \propagator{\mathbf{t}'_2}{0} \ketbra{\phi}{\phi} \propagator{0}{\mathbf{t}_2} &= W_2 \propagator{\mathbf{s}_2}{0} \ketbra{\phi}{\phi} \propagator{0}{\mathbf{s}_2} V_2^\dagger,
    \end{align}
    for \(V_j\), \(W_j\), and \(\mathbf{s}_j\) as appearing in the lemma.
    Due to the \(\delta\)-functions \(\delta_{\mathbf{t}_{1,B} \mathbf{t}'_{1,B}} \delta_{\mathbf{t}_{2,B} \mathbf{t}'_{2,B}}\), we must have that \(V_j\) and \(W_j\) are supported within \(A\), as all the propagators \(\propagator{t}{t'}\) appearing in their definition are the identity when \(t = t'\). Furthermore, the \(\delta\)-function \(\delta_{\mathbf{t}_{1,A} \mathbf{t}_{2,A}} \delta_{\mathbf{t}'_{1,A} \mathbf{t}'_{2,A}}\) then demands that \(V_1 = V_2\) and \(W_1 = W_2\). Substituting this information and canceling \(V_1 V_1^\dagger\) and \(W_1 W_1^\dagger\), the left hand side of \eqnref{eqn:FactorBound} becomes
    \begin{equation}\label{eqn:NoBondPartial}
        \Tr \left[\Tr_{B}[\propagator{\mathbf{s}_1}{0} \ketbra{\phi}{\phi} \propagator{0}{\mathbf{s}_1}]
        \Tr_{B}[\propagator{\mathbf{s}_2}{0} \ketbra{\phi}{\phi} \propagator{0}{\mathbf{s}_2}] \right],
    \end{equation}
    where we must have \(\mathbf{s}_{1A} = \mathbf{s}_{2A} = \mathbf{s}_A\) from the \(\delta_{\mathbf{t}_{1,A} \mathbf{t}_{2,A}} \delta_{\mathbf{t}'_{1,A} \mathbf{t}'_{2,A}}\) constraint. Case~\ref{itm:NoLastBond} is completed by showing that \(\mathbf{s}_j = (\mathbf{s}_{jB1}, \mathbf{s}_{jA}, \mathbf{s}_{jB2})\) in \eqnref{eqn:NoBondPartial} can be replaced by just \(\mathbf{s}_{jA}\). Indeed, the propagator decomposes as
    \begin{equation}
        \propagator{\mathbf{s}_{jB1}}{0} \propagator{\mathbf{s}_{jA}}{0} \propagator{\mathbf{s}_{jB2}}{0}.
    \end{equation}
    The first factor can be cycled through the partial trace and canceled with its conjugate, as it is entirely contained within \(B\). The last factor can be replaced by backwards evolution when acting on \(\ket{\phi}\). Indeed, considering a single hand we have
    \begin{equation}
        \propagator{s_M}{0} \ket{\phi} = \propagator{s_M}{n} \propagator{n}{0} \ket{\phi} = e^{i \phi} \propagator{s_M}{n} \ket{\phi},
    \end{equation}
    and we observe that \(\propagator{s_M}{n}\) has disjoint support from all the other propagators in \(\propagator{\mathbf{s}}{0}\). Iterating this, we have
    \begin{equation}
        \propagator{\mathbf{s}_{jB2}}{0} \ket{\phi} \propto \propagator{\mathbf{s}_{jB2}}{n} \ket{\phi},
    \end{equation}
    where ``\(\propto\)'' is equality up to a phase and
    \begin{equation}
        \propagator{\mathbf{s}_{jB2}}{n} = \propagator{s_{jM}}{n} \cdots \propagator{s_{jk}}{n}
    \end{equation}
    is also supported inside \(B\). Thus, it can similarly be canceled inside the partial trace. We are left with
    \begin{equation}
        \Tr \left[\Tr_{B}[\propagator{\mathbf{s}_A}{0} \ketbra{\phi}{\phi} \propagator{0}{\mathbf{s}_A}]^2 \right] = e^{- \SecRenyi{A}{\propagator{\mathbf{s}_A}{0} \ket{\phi}}},
    \end{equation}
    as desired.

    Case~\ref{itm:YesLastBond} proceeds essentially identically, except that the parts of \(\propagator{\mathbf{s}}{0}\) which cannot be canceled are now the parts in \(B\), instead of \(A\).
    Indeed, we can simply use that the entropy of a subsystem is identical to that of its complement for a pure state, \(\SecRenyi{A}{\ket{\Psi}} = \SecRenyi{B}{\ket{\Psi}}\), to interchange the role of \(A\) and \(B\) in the above calculations.
    Then we get a bound involving \(\SecRenyi{B}{\propagator{\mathbf{s}_B}{0}\ket{\phi}}\), which can again be replaced with \(\SecRenyi{A}{\propagator{\mathbf{s}_B}{0}\ket{\phi}}\).
    The statement of the theorem minimizes over all hand positions \(\mathbf{s}\), whether they are contained within \(A\) or \(B\), so this is sufficient for the claimed result.
    \end{proof}

    It is also possible to obtain a loose lower bound on the free-fermion contribution to the second R\'eyni entropy.

    \begin{lemma}\label{lem:FermionEntropyBound}
        The second R\'eyni entropy of any \(M\)-fermion plane-wave state obeys
        \begin{equation}
            \SecRenyi{A}{c^\dagger_{\mathbf{k},\phi}\ket{\Omega}} \geq \frac{\ell^2}{n^2}
        \end{equation}
        for any interval \(A\) of length \(\ell \leq n/2\).
    \end{lemma}
    \begin{proof}
        We use the result from Ref.~\cite{Muth2011entropybound} that $S^{(2)}_A \geq 2\Delta N_A^2$, where \(\Delta N_A^2\) is the number variance in the subsystem \(A\). 
        Let $\langle O \rangle_{\mathbf{k},\phi} := \bra{\Omega} c_{\mathbf{k},\phi} O c^\dagger_{\mathbf{k},\phi}\ket{\Omega}$; we have 
        \begin{align}
            \Delta N_A^2 
            & := \left\langle \left(\sum_{t\in A} c^\dagger_t c_t\right)^2 \right\rangle_{\mathbf{k},\phi} - \left\langle \sum_{t\in A} c^\dagger_t c_t \right\rangle_{\mathbf{k},\phi}^2 \nonumber \\
            & = \sum_{t,t'\in A} \langle c^\dagger_t c_t c^\dagger_{t'} c_{t'} \rangle_{\mathbf{k},\phi} - (\ell M/n)^2.
        \end{align}
        Using the change of basis $c_t := n^{-1/2} \sum_q e^{i(q-\phi/n)t} c_{q,\phi}$ (with $q \in \tfrac{2\pi}{n}\{0,1, ...,n-1\}$) we get 
        \begin{multline}
            \langle c^\dagger_t c_t c^\dagger_{t'} c_{t'} \rangle_{\mathbf{k},\phi}
            = \frac{1}{n^2} \sum_{q_1^{(\prime)},q_2^{(\prime)}} e^{i(q_1-q_2)t + i(q_1'-q_2')t'} \\
            \times\langle c^\dagger_{q_1,\phi} c_{q_2,\phi} c_{q_1',\phi}^\dagger c_{q_2',\phi} \rangle_{\mathbf{k},\phi}.
        \end{multline}
        Then using Wick's theorem and $\langle c^\dagger_{q,\phi} c_{q,\phi}\rangle_{\mathbf{k},\phi} = \delta_{q\in\mathbf k}$, 
        \begin{align}
            \langle c^\dagger_t c_t c^\dagger_{t'} c_{t'} \rangle_{\mathbf{k},\phi}
            & = \frac{M^2}{n^2} + \frac{1}{n^2} \sum_{q\in\mathbf k} \sum_{q'\notin\mathbf k} e^{i(q-q')(t-t')}
        \end{align}
        and thus, assuming without loss of generality that $A = \{0,1, ..., \ell-1\}$, 
        \begin{align}
            \Delta N_A^2 
            & = \frac{1}{n^2} \sum_{q\in\mathbf k} \sum_{q'\notin\mathbf k} \left| \sum_{t\in A}e^{i(q-q')t} \right|^2 \nonumber \\
            & = \frac{1}{n^2} \sum_{q\in\mathbf k} \sum_{q'\notin\mathbf k}  \left| \frac{e^{i(q-q')\ell} - 1}{ e^{i(q-q')} - 1} \right|^2 \nonumber \\
            & = \frac{1}{n^2} \sum_{q\in\mathbf k} \sum_{q'\notin\mathbf k} \frac{\sin^2[(q-q')\ell/2]}{\sin^2[(q-q')/2]}.
        \end{align}
        To get a loose lower bound, we can restrict the sum to the boundary of the ocuupied states: pairs $q,q'$ such that $|q-q'|=2\pi/n$. There are at least two such pairs provided that $0<M<n$. This gives
        \begin{align}
            \Delta N_A^2 
            & \geq \frac{2}{n^2} \frac{\sin^2[\pi \ell/n]}{\sin^2[\pi/n]} \geq \frac{\ell^2}{2n^2}
        \end{align}
        when $\ell \leq n/2$---we can upper-bound the denominator by $(\pi/n)^2$ and lower-bound the numerator by $(\pi\ell/n)^2/4$ in the relevant range. 
    \end{proof}

    It should be possible to substantially improve Lemma~\ref{lem:FermionEntropyBound}---the lowest entropy state at finite filling fraction should be the Fermi sea, which has logarithmic-in-\(\ell\) entanglement~\cite{Calabrese2004qft,Wolf2006fermionlogentropy,Gioev2006freefermionentropy,Calabrese2009cft}, independent of \(n\). However, Lemma~\ref{lem:FermionEntropyBound} is already sufficient for our purposes, as we only compare the clock entanglement to the exponentially small (in \(n\)) differences between the spin system entanglement and its maximum value \(\ell \ln(2)\).

\section{Review of LFSRs}
\label{app:LFSRReview}

Here we provide a brief review of the mathematical tools necessary to understand LFSRs, specifically the algebra of finite fields~\cite{Lidl_1996_finitefields}.

First we recall that $\mathbb{F}_2$ is the binary field---the set $\{0,1\}$ equipped with addition and multiplication. This meets the axioms of a field (commutativity, associativity, distributivity, existence of additive and multiplicative identity elements, existence of additive and multiplicative inverses). Crucially the field has characteristic 2, i.e., $1+1 = 0$. 

We can then define the finite field $\mathbb{F}_{2^n}$ as follows. 
Consider $\mathbb{F}_2[x]$, the ring of polynomials in one variable $x$ with coefficients in $\mathbb{F}_2$. This is a ring, not a field, as multiplication is not invertible. Next, take a polynomial $p(x)$ of degree $n$ and consider the quotient $\mathbb{F}[x] / \langle p(x) \rangle$, comprising equivalence classes of polynomials modulo multiples of $p(x)$: $q_1(x)\equiv q_2(x)$ if and only if $q_1(x) - q_2(x) = p(x) r(x)$ with $r(x) \in \mathbb{F}[x]$. 
There are $2^n$ such equivalence classes, each represented by a polynomial of degree $\leq n-1$:
\begin{equation}
    \mathbb{F}_2[x]/ \langle p(x) \rangle 
    \ni [c_0 + c_1 x + \dots c_{n-1} x^{n-1}],
    \quad c_i \in \mathbb{F}_2. \label{eq:app_f2n_repr}
\end{equation}
A standard result in algebra states that the quotient $\mathbb{F}_2[x]/p(x)$ forms a field if and only if $p(x)$ is {\it irreducible}, i.e., is not divisible by any polynomials other than $1$ and itself. This is the finite field of $2^n$ elements, $\mathbb{F}_{2^n}$ [any irreducible $p(x)$ gives the same result up to isomorphism]. This construction is analogous to how one builds the finite field $\mathbb{Z}_p$ as a quotient of $\mathbb{Z}$ modulo a prime number $p$; irreducibility of $p(x)$ plays a role analogous to primality of $p$ in ensuring the existence of multiplicative inverses.

A degree-$n$ polynomial $p(x) \in \mathbb{F}_2[x]$ is {\it primitive} if (i) it is irreducible, and (ii) the minimum nonzero $l$ such that $x^l \equiv 1 \mod p(x)$ is $l = 2^n-1$. Condition (ii) can be rephrased as stating that all $2^n-1$ non-zero elements of the field can be generated as powers of $[x]$: 
$\{ [x]^l\}_{l=1}^{2^n-1} = \mathbb{F}_{2^n}\setminus\{[0]\}$. 

Primitive polynomials over $\mathbb{F}_2$ play an important role in the theory of maximal LFSRs~\cite{Tausworthe1965random,Golomb_2017_shiftregisterbook}. Consider forming $\mathbb{F}_{2^n} = \mathbb{F}_2[x] / \langle p(x) \rangle$ with $p(x) = 1 + x^{n} + \sum_{i=1}^{n-1} a_{i} x^i$ a primitive polynomial of degree $n$. Note that the coefficient of $x^{n}$ is fixed to 1 because $p$ has degree $n$, and the coefficient of $x^0$ is fixed to $1$ because $p$ is irreducible (thus not divisible by $x$). 
Now consider the map $[q(x)] \mapsto [x]\cdot [q(x)] \equiv [xq(x)]$ of $\mathbb{F}_{2^n}$ into itself. 
In the representation of each element of $\mathbb{F}_{2^n}$ as a binary vector of coefficients $(c_i)_{i=0}^{n-1}$, Eq.~\eqref{eq:app_f2n_repr}, we can write the action of this map as
\begin{equation}
    [x]\left[ \sum_{i=0}^{n-1} c_i x^i \right]
    = \left [\sum_{i=0}^{n-1} c_i x^{i+1} \right]
    = \left [\sum_{i=0}^{n-1} c_i' x^{i} \right],
\end{equation}
defining new coefficients $(c'_i)_{i=0}^{n-1}$. To obtain the new coefficients, note that multiplication by $x$ generates a degree-$n$ term, $x^n$, that can be eliminated by using the equivalence relation $p(x) \equiv 0$ as $x^n \equiv 1 + \sum_{i=1}^{n-1} a_{i} x^i$, giving
\begin{equation}
c'_i = 
\left\{ 
    \begin{array}{l l}
    c_{n-1} \quad & \text{for }i=0, \\
    c_{i-1} + c_{n-1} a_i \quad & \text{for } i>0. 
    \end{array}
\right.
\end{equation}
This rule is known as a LFSR in {\it Galois form}. It corresponds to a circuit like the one in \autoref{fig:lfsr}(a) up to exchanging the control and target qubits in each CNOT gate (i.e., conjugating by a Hadamard gate on each qubit).
The more commonly used form of LFSRs, which we adopt in \autoref{sec:LFSRModel}, is known as {\it Fibonacci form}. The two are equivalent, being related by a unitary transformation (global Hadamard). 
Note our convention in \autoref{sec:LFSRModel} also applies a spatial inversion, i.e., bits shift backward rather than forward---this is so that the resulting staircase circuit [\autoref{fig:lfsr}(a)] has the same layout as the one used in \autoref{sec:FKClockModel}. 

This construction of LFSRs from multiplications in $\mathbb{F}_{2^n}$ allows us to use algebraic tools to study their orbits. 
In particular, the sequence of states of the LFSR starting from $[1]$ (equivalence class of the constant polynomial) and iterating the update $l$ times is given by $\{[x]^l\}$; 
if $p(x)$ is primitive, then by definition this contains all of $\mathbb{F}_{2^n}\setminus \{[0]\}$, i.e., all non-zero states of the LFSR. Thus the LFSR is maximal (contains an orbit of size $2^n-1$). 
As a result, the search for $n$-bit maximal LFRSs reduces to the search for primitive polynomials of degree $n$ over $\mathbb{F}_2$. 

Primitive polynomials of any degree $n$ can be constructed efficiently [in time $O(n)$]~\cite{rifa_1995_fast}. 
The number of primitive polynomials of degree $n$ can be expressed in terms of {\it Euler's totient function} $\varphi$ as $\varphi(2^n-1) / n$. In general $\varphi$ is bounded as 
\begin{equation}
    \frac{x}{c \log \log x} \leq \varphi(x) \leq x
\end{equation}
for some constant $c$ at sufficiently large argument $x$. It follows that fraction of choices for the coefficients $\{a_i\}$ such that the resulting LFSR is maximal is $\Omega(1/n\log n)$ when $n$ is large. 

We conclude by mentioning two more results on primitive polynomial that are relevant to our physical application.
First, Theorem 1 in Ref.~\cite{Cohen_2004_primitivepoly} states that it is always possible to find a primitive polynomial of degree $n$ in $\mathbb{F}_2[x]$ with its first $n/4$ coefficients $a_1, ..., a_{n/4}$ fixed to arbitrary values. This implies the ability to grow any given $n$-bit maximal LFSR into a $4n$-bit maximal LFSR, thus taking a thermodynamic limit along a family of system sizes $n_k = n_0 4^k$.  
Secondly, it is known that many (though not all) system sizes $n$ admit {\it primitive trinomials}~\cite{Zierler_1968_primitivetrinomial}, $p(x) = x^n + x^k + 1$ for some $k$. These correspond to LFSR circuits with a {\it single} \(\mathsf{CNOT}\) gate, at position $k$. 
Furthermore, every $n\leq 660$ admits a maximal LFSR with at most three $\mathsf{CNOT}$ gates~\cite{Rajski2003primitive,Kim2025catalyticzrotationsconstanttdepth}.
Plugged into our periodic FK clock construction, this gives Hamiltonians that are nearly translationally invariant (where all but a constant number of local Hamiltonian terms are identical up to translation). 

\section{Proof of Lemma~\ref{lemma:pauli}} \label{app:proof}

Here we prove Lemma~\ref{lemma:pauli}, which is one of the main technical results of this work, underpinning the construction of provably thermal Floquet eigenstates in \autoref{sec:LFSRModel}.

Let us take a Pauli operaor $P = X_{\mathbf a} Z_{\mathbf b}$ (we neglect phase factors), where $X_{\mathbf u} = \bigotimes_{i=1}^n X_i^{u_i}$ and $Z_{\mathbf v} = \bigotimes_{i=1}^n Z_i^{v_i}$, $\mathbf{u},\mathbf{v} \in \{0,1\}^n$. We have 
    \begin{align}
        \bra{\psi_q} X_{\mathbf u} Z_{\mathbf v} \ket{\psi_{q'}}
        & = \sum_{\substack{\mathbf z \neq \boldsymbol{0} \\ \mathbf z' \neq \boldsymbol{0}}}  \frac{ \bra{\mathbf z} X_{\mathbf u} Z_{\mathbf v} \ket{\mathbf z'} }{2^n-1} \omega^{q'\log_\alpha(\mathbf z')- q\log_\alpha(\mathbf z)} \nonumber \\
        & = \sum_{\mathbf z' \neq \boldsymbol{0},\mathbf u} \frac{(-1)^{\mathbf v\cdot \mathbf z'} }{2^n-1} \omega^{q'\log_\alpha(\mathbf z')- q\log_\alpha(\mathbf z'+\mathbf u)} \label{eq:app_pauli_me}
    \end{align}
where we used $Z_{\mathbf v} \ket{\mathbf z'} = (-1)^{\mathbf v \cdot \mathbf z'} \ket{\mathbf z'}$ and $\bra{\mathbf z} X_{\mathbf u} \ket{\mathbf z'} = \delta_{\mathbf z, \mathbf z'+\mathbf u}$. 
    
Now we invoke the following properties of the discrete logarithm: 
\begin{itemize}
\item[(i)] $r \log_\alpha(\mathbf z) = \log_\alpha(\mathbf z^r) \text{ mod } 2^n-1$ for all $r\in \mathbb{Z}_{2^n-1}$;
\item[(ii)] $\log_\alpha(\mathbf x) - \log_\alpha(\mathbf y) = \log_\alpha(\mathbf x \mathbf y^{-1}) \text{ mod } 2^n-1$.
\end{itemize}
Here the power $\mathbf z^r$ and the multiplicative inverse $\mathbf y^{-1}$ are both defined in the finite field $\mathbb{F}_{2^n}$ (Appendix~\ref{app:LFSRReview}). 
Using these facts, we may rewrite 
\begin{align}
    \omega^{q'\log_\alpha(\mathbf z) - q \log_\alpha(\mathbf z+\mathbf u)}
    & = \omega^{\log_\alpha(\mathbf z^{q'}) - \log_\alpha[(\mathbf z + \mathbf u)^q]} \nonumber \\
    & = \omega^{\log_\alpha\left[\mathbf z^{q'} (\mathbf z + \mathbf u)^{-q}\right]}.
\end{align}

At this point it is helpful to define the following functions:
\begin{align}
\varphi(\mathbf z) & = (-1)^{\mathbf v \cdot \mathbf z}, \\
\chi(\mathbf z) & = \omega^{\log_\alpha(\mathbf z)}, \\ 
g(\mathbf z) & = \mathbf{z}^{q'} (\mathbf z + \mathbf u)^{-q}. 
\end{align}
Here $\varphi$ is an {\it additive character} for $\mathbb{F}_{2^n}$, i.e. it satisfies 
\begin{equation} \varphi(\mathbf{z} + \mathbf{z}') = \varphi(\mathbf{z}) \varphi(\mathbf{z}'),
\end{equation}
while $\chi$ is a {\it multiplicative character}, i.e. it satisfies 
\begin{equation} \chi(\mathbf{z} \mathbf{z}') = \chi(\mathbf{z}) \chi(\mathbf{z}').
\end{equation} 
The latter follows from property (ii) of the discrete logarithm above. 
$g$ is a rational function with a $q$-fold pole at $\mathbf z = \mathbf u$ (recall the field $\mathbb{F}_{2^n}$ has characteristic 2, so $-\mathbf{u} = + \mathbf{u}$) and a $(q'-q)$-fold pole at infinity, assuming without loss of generality that $q' \geq q$. 

In terms of these functions, our Pauli matrix element Eq.~\eqref{eq:app_pauli_me} reads
\begin{equation}
    \bra{\psi_q} X_{\mathbf u} Z_{\mathbf v} \ket{\psi_{q'}} = \frac{1}{2^n-1} \sum_{\mathbf z\in \mathbb{F}_{2^n}\setminus \{\boldsymbol{0},\mathbf u\}} \varphi(\mathbf z) \chi(g(\mathbf z)).
\end{equation}
Sums of this type are known as {\it mixed character sums}, as they combine additive and multiplicative characters. They are widely studied in algebraic number theory where they arise in connection with the Riemann hypothesis~\cite{Weil_1948_exponentialsums}. 
Mixed character sums often obey bounds of the form
\begin{equation}
    \left| \sum_{\mathbf z\in \mathbb{F}_{2^n} \setminus\{ \boldsymbol{0}\}}  \varphi(\mathbf z) \chi(\mathbf z) \right| \leq O(2^{n/2}), \label{eq:mixedchar_generic}
\end{equation}
which is the same bound one gets with high probability when summing $2^n$ random phases. This shows that the phase factors $\chi(\mathbf z) = \omega^{\log_\alpha \mathbf z}$ have ``random-like'' correlations with all linear subspaces of bitstrings, specified by $\varphi(\mathbf z) = (-1)^{\mathbf v\cdot \mathbf z}$.

Many bounds on mixed character sums along the lines of Eq.~\eqref{eq:mixedchar_generic} exist.
Here we will use the following version\footnote{
The original theorem allows for sums on arbitrary algebraic curves $X$ and involves a term $2g_X$ on the right hand side, with $g_X$ the genus of $X$; in our case $X$ is the projective line $X = \mathbb{F}_{2^n} \cup \{ \infty \}$, which has genus $g_X = 0$.}:
\begin{thm}[Bound on mixed character sums, adapted from Ref.~\cite{castro_2000_mixedsums} Theorem 13] \label{thm:mixed_character}
Given a nontrivial additive character $\varphi$, nontrivial multiplicative character $\chi$, and rational functions $f,g$ on $\mathbb{F}_{2^n}$, we have
    \begin{align}
        \left| \sum_{\mathbf z \in \mathbb{F}_{2^n} \cup \{\infty\} } {}^{'} \varphi(f(\mathbf z)) \chi(g(\mathbf z)) \right| 
        \leq C 2^{n/2}
    \end{align}
    where the primed sum excludes any zeros of $g$ and poles of $f$ or $g$, and the constant $C$ is given by
    \begin{align}
        C =  s + l + d - r - 2.
    \end{align}
    Here $s$ is the total number of points that are either poles or zeros of $g$, $l$ is the number of points that are poles of $f$, $d$ is the total order of the poles of $f$, and $r$ is the number of points that are simultaneously poles of $f$ and zeros or poles of $g$. 
\end{thm}

Since in our case $f(\mathbf z) = \mathbf z$, we have $l = 1$ (a single pole at $\mathbf z= \infty$) and $d = 1$ (the order of the pole is 1).
For the coefficients $s$ and $r$ we need to distinguish two cases (recall we assumed without loss of generality $q'\geq q$): 
\begin{itemize} 
\item $q = q'$: $g(\mathbf z) = [\mathbf z / (\mathbf z + \mathbf u)]^q$ does not have a pole at infinity. We have $s = 2$ (a zero at $\mathbf z = \boldsymbol{0}$ and a pole at $\mathbf z = \mathbf u$, the multiplicity $q$ does not count) and $r = 0$ (no poles in common between $f$ and $g$). 
\item $q < q'$: $g(\mathbf z)$ has a pole at infinity. We have $s = 3$ (a zero at $\mathbf z = \boldsymbol{0}$, poles at $\mathbf z = \mathbf u$ and $\mathbf z = \infty$) and $r = 1$ (the pole at infinity is common to $f$ and $g$). 
\end{itemize}
Either way, we obtain $s-r = 2$, and thus $C = 2$:
\begin{equation}
    |\bra{\psi_q} X_{\mathbf u} Z_{\mathbf v} \ket{\psi_{q'}}| \leq  \frac{C2^{n/2}}{2^n-1} = \frac{2^{1+n/2}}{2^n-1}
\end{equation}
which concludes the proof for the case of $\mathbf u, \mathbf v \neq \boldsymbol{0}$ (where both $\varphi$ and $\chi$ are nontrivial characters and Theorem~\ref{thm:mixed_character} applies). In the case where either character is trivial the calculation can be carried out explicitly, see App.~\ref{app:jacobisums}, and the results also obey the bound.
    
\section{Details on calculations of Pauli matrix elements}
\label{app:PauliDetails}

\subsection{Construction of eigenstates}

Here we provide details on the method used to obtain the data in \autoref{fig:lfsr}(b-c). 
The LFSRs used for each system size are specified by the coefficients in Table~\ref{tab:taps}. 
For each one we generate the maximal orbit $\{ \alpha^j\boldsymbol{1}:\ j=0,\dots 2^n-2\}$ by starting from $\boldsymbol{1} = 00\dots 01$ and iterating the update $\mathbf z' = \alpha \mathbf z$ until the orbit closes; we verify the orbit length is indeed $2^n-1$. 
We then numerically create all nontrivial eigenstates one component at a time as 
\begin{equation}
    \braket{\alpha^j\boldsymbol{1}}{\psi_q} = \frac{\omega^{qj}}{\sqrt{ 2^n-1}} 
\end{equation}
with $\omega = e^{2\pi i / (2^n-1)}$, for $j \in \{0, ..., 2^n-2\}$ and $q \in \{1, ..., 2^{n-2}\}$. 

\subsection{Pauli orbits}

Due to the Clifford nature of the LFSR circuit, the $4^n-1$ nontrivial Pauli operators break up into orbits, $\{ U_{n:0}^{-l} P U_{n:0}^l\}_l$. It is easy to see that each orbit has length $2^n-1$, same as the orbit of bitstring states---for example one can show that
\begin{align}
    \bra{\mathbf z} U_{n:0}^\dagger X_{\mathbf u} U_{n:0} \ket{\mathbf z'}
    & = \bra{\alpha \mathbf z}  X_{\mathbf u} \ket{\alpha \mathbf z'}
    = \delta_{\alpha \mathbf z, \alpha \mathbf z' + \mathbf u} \nonumber \\
    & = \delta_{\mathbf z, \mathbf z' + \alpha^{-1} \mathbf u} 
    = \bra{\mathbf z}  X_{\alpha^{-1} \mathbf u} \ket{\mathbf z'},
\end{align}
implying $U_{n:0}^\dagger X_{\mathbf u} U_{n:0} = X_{\alpha^{-1} \mathbf u}$; since $\alpha^{-1}$ has the same orbits as $\alpha$ (traversed in inverse order), this shows that the $X$ part of any Pauli string has periodicity $2^n-1$. The same conclusion applies to the $Z$ part:
\begin{align}
    \bra{\mathbf z} U_{n:0}^\dagger Z_{\mathbf v} U_{n:0} \ket{\mathbf z'}
    & = \bra{\alpha \mathbf z}  Z_{\mathbf v} \ket{\alpha \mathbf z'}
    = (-1)^{\mathbf v \cdot (\alpha \mathbf z)} \delta_{\alpha \mathbf z, \alpha \mathbf z'} \nonumber \\
    & = (-1)^{(\alpha^T \mathbf v) \cdot \mathbf z} \delta_{\mathbf z,\mathbf z'}
    = \bra{\mathbf z} Z_{\alpha^T \mathbf u} \ket{\mathbf z'},
\end{align}
which implies $U_{n:0}^\dagger Z_{\mathbf v} U_{n:0}  = Z_{\alpha^T \mathbf v}$; $\alpha^T$ represents the circuit obtained by a spatial inversion and global Hadamard which has the same orbit structure as $\alpha$. 
So the $4^n-1$ nontrivial Pauli strings split into $2^n+1$ orbits of length $2^n-1$.
Up to a phase, we can reduce Pauli matrix elements between eigenstates to a canonical representative of each Pauli orbit:
\begin{equation}\label{eqn:PauliOrbits}
    \bra{\psi_q} X_{\mathbf u} Z_{\mathbf v} \ket{\psi_{q'}}  = 
    \left\{ 
    \begin{array}{l l}
    \omega^{l(q-q')} \bra{\psi_q} Z_{\boldsymbol{1}} \ket{\psi_{q'}} & \text{ if } \mathbf u = \boldsymbol{0},  \\
    \omega^{l(q-q')} \bra{\psi_q} X_{\mathbf 1} Z_{\mathbf v'} \ket{\psi_{q'}} &  \text{ otherwise},
    \end{array}
    \right.
\end{equation}
for some integer $l$. The $\mathbf u = \boldsymbol{0}$ specifies one orbit, while $\mathbf u \neq \boldsymbol{0}$ specifies the remaining $2^n$ orbits, one for each $\mathbf{v}'$.
Thus, up to phase factors, there are only $2^n+1$ inequivalent Pauli operators for the purpose of computing matrix elements.

\begin{table}
    \centering
    \begin{tabular}{c|c}
    \toprule
    $n$ & $\{ i:a_i=1\}$ \\ \midrule 
    9 & $\{4\}$ \\
    10 & $\{3\}$ \\
    11 & $\{2\}$ \\
    12 & $\{ 1,2,8 \}$  \\
    13 & $\{ 1, 2, 5\}$  \\
    14 & $\{ 1, 2, 13 \}$ \\
    15 & $\{1\}$ \\
    16 & $\{1, 3, 12\}$  \\
    17 & $\{3\}$ \\
    18 & $\{7\}$ \\ \bottomrule 
    \end{tabular}
    \caption{Maximal $n$-bit LFSRs used in Fig.~\ref{fig:lfsr}. For each $n$ we list the values of $i$ where $a_i = 1$ (all other entries of $a$ are 0). We follow the convention of Eq.~\eqref{eq:lfsr}. }
    \label{tab:taps}
\end{table}

\subsection{Special cases}\label{app:jacobisums}

For certain Pauli operators the calculation of matrix elements is analytically tractable. These Pauli operators account for the visible spike in the data of \autoref{fig:lfsr}(c) at $|\mu| = 1/2$. 

First, consider Pauli operators $P = Z_{\mathbf v}$, made only of $Z$ Pauli matrices. 
Matrix elements of these operators between nontrivial LFSR eigenstates can be written as
\begin{equation}
    |\bra{\psi_q} Z_{\mathbf v} \ket{\psi_{q'}}| = |\bra{\psi_q} Z_{\boldsymbol{1}} \ket{\psi_{q'}}| 
\end{equation}
using the Pauli orbit structure discussed previously to choose a convenient representative [\eqnref{eqn:PauliOrbits}].
Then we can write
\begin{align}
    (2^n-1) \bra{\psi_q} Z_{\boldsymbol{1}} \ket{\psi_{q'}} 
    & = \sum_{\mathbf z\neq \boldsymbol{0}} \chi^{q-q'}(\mathbf z) \varphi(\mathbf z),
\end{align}
where we introduced the multiplicative character $\chi(\mathbf z) = \omega^{\log_\alpha(\mathbf z)}$ and the additive character $\varphi(\mathbf z) = (-1)^{\boldsymbol{1}\cdot \mathbf z} = (-1)^{z_{n-1}}$. 
\begin{itemize}
    \item If $q=q'$, the sum involves only the additive character, and we get 
    \begin{equation}
        (2^n-1) \bra{\psi_q} Z_{\boldsymbol{1}} \ket{\psi_{q}} =  \left( \sum_{\mathbf z} \varphi(\mathbf z) \right) - 1 = -1;
    \end{equation}
    \item If $q\neq q'$, then $\chi^{q-q'}$ is a nontrivial multiplicative character, and the sum is known as a {\it Gauss sum} $G(\chi^{q-q'})$. 
    A standard result in finite field algebra states that, for any nontrivial multiplicative character, the absolute value of the Gauss sum is fixed to~\cite{Ireland1982gaussjacobi}
    \begin{equation}
        |G(\chi)| = 2^{n/2}
    \end{equation}
    (generally $\sqrt{q}$ in $\mathbb{F}_q$). It follows that
    \begin{equation}
    |\bra{\psi_q} Z_{\boldsymbol{1}} \ket{\psi_{q'}}| = \frac{2^{n/2}}{2^n-1}. 
    \end{equation}
\end{itemize}

Next, consider Pauli operators $P = X_{\mathbf u}$, made only of $X$ Pauli matrices. 
We have 
\begin{align}
    \bra{\psi_q} X_{\boldsymbol{1}} \ket{\psi_{q'}} 
    & = \frac{1}{2^n-1} \sum_{\mathbf z \neq \boldsymbol{0}}\sum_{\mathbf z'\neq \boldsymbol{0}} \chi^q (\mathbf z) \chi^{-q'} (\mathbf z') \delta_{\mathbf z,\mathbf z'+\boldsymbol{1}} \nonumber \\
    & = \frac{1}{2^n-1} \sum_{\mathbf z \neq \boldsymbol{0},\boldsymbol{1}} 
    \chi^q (\mathbf z) \chi^{-q'} (\boldsymbol{1} + \mathbf z) \nonumber \\
    & = \frac{1}{2^n-1} J(\chi^q, \chi^{-q'}).
\end{align}
Here we introduced the {\it Jacobi sum} defined for two multiplicative characters $\chi_{1,2}$ as 
\begin{equation}
    J(\chi_1, \chi_2) = \sum_{\mathbf z\neq \boldsymbol{0}, \boldsymbol{1} } \chi_1(\mathbf z) \chi_2(\boldsymbol{1} + \mathbf z).
\end{equation}
Provided each of $\chi_1$, $\chi_2$, and $\chi_1 \chi_2$ are nontrivial, then another standard result in finite-field algebra relates the Jacobi sum to Gauss sums~\cite{Ireland1982gaussjacobi}:
\begin{equation}
    J(\chi_1,\chi_2) = \frac{G(\chi_1) G(\chi_2)}{G(\chi_1 \chi_2)}, \label{eq:jacobi_sum}
\end{equation}
which in particular implies $|J(\chi_1,\chi_2)| = 2^{n/2}$. 
\begin{itemize}
    \item If $q = q'$, then $\chi^q \chi^{-q'}$ is trivial (the identity character) and we cannot use Eq.~\eqref{eq:jacobi_sum}. However in that case we can compute the sum directly:
    \begin{align}
        \sum_{\mathbf z\neq \boldsymbol{0},\boldsymbol{1}} \chi^{-q}(\mathbf{z}) \chi^q(\boldsymbol{1} + \mathbf{z})
        & = \sum_{\mathbf z\neq \boldsymbol{0},\boldsymbol{1}} \chi^q(\boldsymbol{1} + \mathbf{z}^{-1}) \nonumber \\
        & = \sum_{\mathbf y \neq \boldsymbol{0},\boldsymbol{1}} \chi^q(\mathbf{y}) = - \chi^q(\boldsymbol{1}),
    \end{align}
    where we changed variable to $\mathbf y = \boldsymbol{1} + \mathbf{z}^{-1}$ in the sum and used the fact that the sum of all roots of unity vanishes: $\sum_{\mathbf y \neq \boldsymbol{0}} \chi^q(\mathbf y) = 0$. Since $- \chi^q(\boldsymbol{1})$ is a phase factor we conclude
    \begin{equation}
        | \bra{\psi_q} X_{\boldsymbol{1}} \ket{\psi_{q}} | = \frac{1}{2^n-1}.
    \end{equation}
    \item If $q\neq q'$ (and $q,q'\neq 0$) then each of $\chi^q$, $\chi^{q'}$, $\chi^{q-q'}$ is nontrivial, and Eq.~\eqref{eq:jacobi_sum} implies 
    \begin{equation}
        | \bra{\psi_q} X_{\boldsymbol{1}} \ket{\psi_{q'}} | = \frac{2^{n/2}}{2^n-1}. 
    \end{equation}
\end{itemize}

Normalized by the upper bound of Lemma~\ref{lemma:pauli}, the diagonal matrix elements (expectation values) of $X_{\mathbf u}$ and $Z_{\mathbf v}$ yield $|\mu| = 2^{-n/2}$, 
while the off-diagonal matrix elements yield $|\mu| = 1/2$. 
The latter explains the ``spike'' clearly visible in \autoref{fig:lfsr}(c) at $|\mu| = 1/2$. Since the special value $|\mu| = 1/2$ is attained in two out of $2^n+1$ Pauli orbits, the size of the feature in the normalized distribution of matrix elements is $\propto 2^{-n}$.

\bibliography{clockmodels}

\end{document}